\documentclass[11pt]{article}
\usepackage{todonotes}
\usepackage[utf8]{inputenc}
\usepackage[english]{babel}
\usepackage{amsfonts}
\usepackage{siunitx}
\usepackage{amssymb}
\usepackage{enumitem}
\usepackage{mathtools}
\usepackage{comment}
\usepackage{xcolor}
\usepackage{caption}
\usepackage{subcaption}
\usepackage{booktabs}	
\usepackage{setspace}
\usepackage{longtable}
\usepackage{graphicx}
\usepackage{fullpage}
\usepackage{soul}
\usepackage{placeins}
\usepackage{tcolorbox}
\usepackage[backref=page,breaklinks=true,hidelinks]{hyperref}
\hypersetup{
    colorlinks,
    linkcolor={red!50!black},
    citecolor={blue!50!black},
    urlcolor={blue!80!black}
}
\renewcommand*{\backref}[1]{}
\renewcommand*{\backrefalt}[4]{%
    \ifcase #1 {\footnotesize(Not cited.)}%
    \or        {\footnotesize(Cited on page~#2.)}%
    \else      {\footnotesize(Cited on pages~#2.)}%
    \fi}
\usepackage{amsthm}
\usepackage{thmtools}
\usepackage{dsfont}
\usepackage{natbib}
\usepackage{mleftright}
\usepackage{xfrac}
\usepackage{rotating}
\usepackage{mdframed}

\presetkeys%
    {todonotes}%
    {inline,backgroundcolor=yellow}{}
    
\newtheorem{asm}{Assumption}

\newtheorem{prop}{Proposition}
\newtheorem{lem}{Lemma}

\newtheorem{defi}{Definition}

\newcommand{\Tr}{\mathcal{T}}
\newcommand{\HO}{\mathcal{H}}
\newcommand{\Lb}{\mathcal{L}}

\newcommand{\0}{\mathbf{0}}
\newcommand{\I}{\mathbb{I}}

\newcommand{\R}{\mathbb{R}}
\newcommand{\N}{\mathcal{N}}
\let\P\relax
\DeclareMathOperator{\P}{P}
\DeclareMathOperator{\E}{E}
\DeclareMathOperator{\Var}{Var}
\DeclareMathOperator{\Cov}{Cov}
\DeclarePairedDelimiter{\Ind}{\mathds{1}\lparen}{\rparen}
\newcommand{\peq}{\mathrel{\phantom{=}}}

\newcommand{\textbox}[1]{
\centering
\vspace{-10pt}
    \begin{tcolorbox}[colframe=black, boxrule=0.5pt, arc=4mm, colback=white!0, width=24em, boxsep=0pt]
        \centering
        {
        \textnormal{
        #1
        }
        }
    \end{tcolorbox}    
\vspace{-10pt}   
}

\usepackage{fvextra}
\DefineVerbatimEnvironment{Verbatim}{Verbatim}{breaklines=true}

\newcommand\nnfootnote[1]{%
  \begin{NoHyper}
  \renewcommand\thefootnote{}\footnote{#1}%
  \addtocounter{footnote}{-1}%
  \end{NoHyper}
}

\graphicspath{{graphs/}}

\title{
    Causal Inference on Outcomes Learned from Text
}
\author{
    Iman Modarressi
    \and Jann Spiess
    \and Amar Venugopal
}

\date{
  This version: March 1, 2025
}

\begin{document}
\renewcommand{\sectionautorefname}{Section}
\renewcommand{\subsectionautorefname}{Section}
\renewcommand{\footnoteautorefname}{Footnote}

\maketitle

\begin{abstract}
    We propose a machine-learning tool that yields causal inference on text in randomized trials. Based on a simple econometric framework in which text may capture outcomes of interest, our procedure addresses three questions: First, is the text affected by the treatment? Second, which outcomes is the effect on? And third, how complete is our description of causal effects? To answer all three questions, our approach uses large language models (LLMs) that suggest systematic differences across two groups of text documents and then provides valid inference based on costly validation. Specifically, we highlight the need for sample splitting to allow for statistical validation of LLM outputs, as well as the need for human labeling to validate substantive claims about how documents differ across groups. We illustrate the tool in a proof-of-concept application using abstracts of academic manuscripts.
\end{abstract}

\nnfootnote{
    Authors are listed in alphabetical order.
    Iman Modarressi (\href{mailto:im560@cam.ac.uk}{im560@cam.ac.uk}), King's College, University of Cambridge;
    Jann Spiess (\href{mailto:jspiess@stanford.edu}{jspiess@stanford.edu}), Graduate School of Business, Stanford University;
    Amar Venugopal (\href{mailto:amarvenu@stanford.edu}{amarvenu@stanford.edu}), Department of Economics, Stanford University.
    A previous version of this manuscript was circulated and presented under the title ``Causal Effects on Text'' (July 2024).
    We thank Amir Feder, Friedrich Geiecke, Akhila Kovvuri, Sendhil Mullainathan, Ashesh Rambachan, Joe Romano, Keyon Vafa, and audience members at the 2024 Philadelphia Fed/Chicago Booth Conference on Frontiers in Machine Learning and Economics and the 2024 Zurich AI + Economics Workshop as well as at Stanford, Yale, Princeton, and Berkeley for helpful comments, suggestions, and discussions.
    We gratefully acknowledge the support of the Stanford Institute for Human-Centered Artificial Intelligence (HAI) and Google Cloud through the 2024 HAI Google Cloud Credits Award program.
}

\section{Introduction}
\label{sec:Introduction}

We run an increasing number of expensive randomized controlled trials (RCTs) to study the causal effect of interventions, yet their analysis remains limited by which outcomes are measured.
While the interventions we study may affect study subjects in many ways, our inference may miss important effects that go beyond the narrow aperture of the outcome variables we have collected or constructed.
For example, a new drug may affect patients' well-being not only in terms of the few health measurements we focus our analysis on and a micro-credit program may affect the financial lives of the poor beyond the investment and consumption data we collect.
Instead, such broader outcomes can be captured by unstructured responses of study subjects, especially in the form of natural language. These responses may encode subtle, context-specific nuances, including emotional tone and emphasis of themes or concepts not anticipated ex ante.
Yet analyzing effects based on such unstructured data requires new tools that combine the complexity of language with the discipline of causal inference.

In this article, we provide a framework for leveraging AI tools to provide valid inference on differences across two groups of text documents, such as the treatment and control groups in a randomized trial.
We first ask \emph{whether} the two groups are systematically different, which we tackle by combining machine-learning predictors and sample splitting. We then conceptualize the task of describing \emph{what} the systematic differences across texts are as a problem of finding and validating simple descriptions of texts in small datasets, which we implement using large language models (LLMs) followed by human validation. Finally, we quantify \emph{how complete} our resulting descriptions are by comparing them against a model-free benchmark.

Conventional econometric tools are limited in making inference on text. A first approach would be to extract a small number pre-specified features. However, such an approach would have to rely on strong priors about which part of the text is affected by a treatment or demonstrates differences across groups, and may miss signal outside of a narrow set of such measures. A second approach could instead use high-dimensional features such as word counts. In this latter case, we may be worried about spurious results when we search through many potential ways that only a few documents may differ, necessitating heavy multiple-testing penalties that may again result in limited power. Furthermore, despite their high dimensionality, features such as word counts may lose crucial contextual information.

Instead of relying on pre-specified text features or simple word counts, we deploy state-of-the-art techniques from artificial intelligence and machine learning in order to extract signals from text, including large language models (LLMs). These approaches can contextualize words in their context and effectively extract signals even when the number of documents we analyze is small, thanks to pre-training on large corpora of texts. At the same time, these tools come with a separate challenge: How can we interpret and ultimately provide statistical guarantees for output extracted from extremely complex black-box machine-learning systems that were furthermore trained on text we do not have access to and refined using human input that we do not observe?

To turn intransparent AI systems into a valid econometric tool, we first translate inference tasks into machine-learning questions and then validate their answers using sample splitting.
We start with asking \emph{whetehr} there is a systematic difference across different groups of text documents. We reframe this task into a question that is more amenable to the use of machine learning: Given a document, can we predict which group it is from better than some trivial baseline? For example, for a text document in a randomized trial, can we classify whether it is from the treatment or control group? If we can do so on a hold-out dataset that was not used to train our machine-learning algorithm, then we can provide evidence for a causal effect through a valid permutation test, without relying on assumptions about the underlying algorithm. This allows us to obtain valid inference even when a large language model provides predictions.

Having established that two groups of text documents vary significantly, we next ask \emph{what} the main differences are. We turn the question of describing differences into the task of generating hypotheses about which themes are different across groups. These themes optimally represent a low-dimensional summary of outcomes captured by text, with the goal of providing a meaningful and quantifiable low-dimensional description of differences.
These themes can include what texts talk about (topics) as well as how texts represent topics (sentiment).
To obtain such features, we provide instructions to a large language model (LLM) to identify themes that describe differences across groups.
For each of the themes, we allow the algorithm to propose a flexible scale for evaluating individual documents, such as the degree to which a topic appears or the tone as a scale that ranges from negative to positive sentiment.
These themes are learned based on training data that includes group assignments, such as which documents come from the treatment group of a randomized trial, to ensure that themes are informative about differences across groups.

To validate our description of differences between text documents, we combine sample splitting and human validation. If we run a complex large language model on a small experimental dataset, we may be worried that it returns theme classifications that do not represent systematic differences and instead are the result of overfitting to the existing data.
Beyond this statistical concern, AI systems like the LLM we deploy do not come with guarantees that the machine-provided summaries are accurate and \emph{substantively} refer to the described concepts. 
We, therefore, provide a sample-splitting scheme that separates (machine) hypothesis generation on the training from (human) hypothesis validation on held-out data. On the training data, our algorithm learns systematic differences across treatment and control groups and suggests ways of scoring themes. On the test data, an expert scores held-out texts according to the proposed themes before group labels became available. As a result, we obtain valid inference on the proposed themes without having to make the assumption that the proposed themes are correct or that an AI would be able to score them correctly, neither of which we think can be guaranteed.

Having obtained and validated descriptions of differences across groups of text documents, we finally ask \emph{how complete} these descriptions are.
To do so, we return to our reverse-prediction approach and compare how well the simple theme scores alone predict group labels relative to an unconstrained benchmark. If our measure of completeness is close to zero, then the themes do not contain a lot of information about systematic differences between text documents, and further analysis may be needed to improve descriptions. If, on the other hand, completeness is close to one, then the simple themes express most of the differences that we can find from the texts in the training sample.

As a proof of concept, we apply our method to abstracts of academic manuscripts on \textit{arXiv}. Starting with a classification into two groups, we analyze their systematic differences using Google's Gemini large language model.
We first demonstrate that these groups are significantly different by showing that the LLM can determine which document comes from which group better than a trivial benchmark.
We then prompt Gemini to describe differences across both in terms of a few core themes.  We then have a member of the research team label the held-out documents according to these themes and validate differences on the hold-out set. Finally, we benchmark how well these themes distinguish between groups based on logistic regression and interpret our results as saying that the intuitive descriptions capture almost all of the differences in our example.
While this proof of concept does not showcase an application to a randomized trial directly, it nevertheless provides suggestive evidence that the method can uncover the kind of systematic differences that allow for causal inferences based on text outcomes.

The approach outlined in this work comes with questions about the replicability of AI-based statistical inference as well as the cost of human validation. We propose a pre-specification strategy for empirical analysis with LLMs based on successive data access, where it is not the algorithm itself but instead its intermediate output that is publicly archived before additional data is accessed and the final validation is performed. Furthermore, we show how recent proposals for combining human and machine labels can be implemented within our workflow, and quantify trade-offs between validation cost and statistical precision in our empirical illustration.

Our work builds upon approaches across the social sciences, computer science, and econometrics on using machine learning to make inferences on text specifically and on high-dimensional outcomes more broadly.
Our approach to detecting effects on text via reverse prediction is building on \cite{Ludwig2017-xk}, while the sample-splitting scheme is similar to \cite{egami2022make}.
The text-scoring approach goes back to \cite{pmlr-v162-zhong22a} and is most closely related to the D5 approach from \cite{zhong2023goaldrivendiscoverydistributional}, while our approach to integrating and validating LLM outputs builds upon \cite{llmapplied} and has parallels to the work on images in \cite{ludwig2024machine}.
The completeness measure is informed by \cite{fudenberg2022measuring}, and our approach to combining machine and human labels follows \cite{robins1994estimation} and \cite{angelopoulos2023prediction}.
We review some related literature in more detail in \autoref{sec:Literature} before laying out our setup in \autoref{sec:Setup}.
We then discuss testing for effects in \autoref{sec:Testing}, describing differences in \autoref{sec:Describing}, and estimating the completeness of descriptions in \autoref{sec:Completeness}.
\autoref{sec:Empirical} provides our empirical illustration using academic abstracts,
\autoref{sec:Replication} discusses replicability and pre-specification, and
\autoref{sec:Conclusion} concludes.

\section{Literature and Related Work}
\label{sec:Literature}

Our project builds upon advances in natural language processing and large language models in computer science, textual analysis in the social sciences, and work in economics, econometrics, and statistics on integrating causal inference and machine learning into causal inference.

\paragraph{Natural language processing and large language models in computer science.}

Natural language processing (NLP) methods have rapidly evolved, enabling automated analysis of textual data. As a foundational model, Latent Dirichlet Allocation \citep[LDA;][]{blei2003latent} is a generative probabilistic model of text as a mixture of a small number of topics, making it possible to uncover hidden thematic structures within large text corpora. The Structural Topic Model \citep[STM;][]{robertssocial} extends topic models by incorporating document-level metadata.
Such topic-modeling methods complement approaches that summarize and differentiate text by classifying documents according to pre-defined categories \citep[e.g.][]{fuchunbayes}.
The emergence of large language models (LLMs) has opened new opportunities for identifying patterns within large texts based on limited task-specific training data. Inspired by \cite{brownfewshot}, which demonstrates that LLMs can learn from examples included in their input (a concept known as in-context learning), \cite{lin2023unlockingspellbasellms} shows that just a few carefully chosen examples in the input prompt can guide pre-trained base models to follow instructions and interact with users comparably to fine-tuned models.

Closely related to the techniques we apply to randomized trials, \citet{pmlr-v162-zhong22a} proposes an LLM-based hypothesis generation procedure in which an LLM is queried to learn descriptions that are true more often for one group of text, relative to another. These hypotheses are then validated via another question--answer AI agent.
Further work by \citet{zhong2023goaldrivendiscoverydistributional} introduces the D5 system, designed to automatically discover differences between extensive text collections based on user-specified goals. This system combines language models to generate hypotheses about differences and validate these hypotheses, leveraging a diagnostic benchmark and a meta-dataset for evaluation. 
Our approach is closely related, particularly since we also use LLMs to uncover systematic differences between groups of documents.
It builds on these approaches by extending the learned descriptions to ``causal themes", implementing verification via ground-truth labeling, and explicitly focussing on causal inference from experimental data.
Related to this line of work, \cite{zhu2022gsclipframeworkexplaining} introduces a framework for explaining changes between images using natural language, followed by an encoding step to quantify differences across groups.

\paragraph{Text analysis in the social sciences.}

The bulk of textual analysis in the social sciences relies upon more traditional NLP methods, 
largely due to their superior interpretability relative to more modern transformer-based deep learning
architectures. \citet{gentzkow_congressional_2018} compiles a corpus of US congressional floor speeches and goes on to parse and analyze the text of these speeches to measure trends in partisanship over time. 
\citet{voelkel_changing_2022} measures the impact of different textual interventions on attitudes towards immigration. 

There has also been a burgeoning literature on textual data methods for causal inference. 
\citet{keith_text_2020} outlines a method for using textual data to remove 
confounding in causal inference settings. \citet{mozer_matching_2020} considers textual data as a covariate in matching estimation of treatment effects. 
\citet{fong_discovery_2016} puts forth a generative model of text used to discover latent topics 
that are particularly relevant for treatment effect estimation. \citet{egami2022make} provides 
a detailed overview of current methods for the use of text in causal inference, including using 
text as treatment, and outlines both the difficulties inherent in using latent representations of text, 
such as topic model proportions, when conducting causal inference. That work builds on \citet{fong_causal_2023} in constructing relevant assumptions to obtain causal estimates from latent features derived from textual data.
\citet{feder2022causal} surveys work on bridging causal inference and natural language processing from the perspective of the NLP literature.

Finally, recent literature has also sought to leverage LLMs for descriptive analysis in textual settings
in the social sciences. \citet{kim_ai-augmented_2024} utilizes LLMs in survey data to impute user responses to unasked questions.
\cite{durvasula2024counting} fine-tunes an LLM to provide high-quality annotations of unstructured text.
Related to our approach, \citet{wu_elicitationgpt_2024} uses LLMs to produce 
summary points of text corpora, which are then used to produce measures of text similarity.
\cite{egami2024using} provides a framework for working with LLM-annotated texts.
\citet{imai2024_llmtreat} extends the use of LLMs to text-as-treatment settings in causal inference, leveraging generative AI to produce sample texts that modify selected features of interest in order to assess causal effects.

\paragraph{Integration of machine learning and AI into causal inference and economics.}

Recent work in economics, econometrics, and statistics has proposed innovative ways of integrating AI and machine learning into applied data anlysis \citep{mullainathan2017machine,athey2018impact,athey2019machine,dell2024deep}.
This includes work that uses supervised machine learning to identify causal effects on high-dimensional outcomes \citep{Ludwig2017-xk}, test for balance between covariates \cite{gagnon2019classification}, and impute missing observations \citep{rambachan2024programevaluationremotelysensed,battaglia2024inference, angelopoulos2023prediction,egami2024using}.
We build upon this work by applying a similar approach to text outcomes, where we use LLMs to predict missing labels.
Furthermore, we adapt the approach of \citet{fudenberg2022measuring} that evaluates the completeness of economic theories to quantify how well our themes describe differences between treatment and control outcomes.

We specifically build upon recent approaches that learn from image and text data in economics. \citet{ludwig2024machine} builds a deep-learning tool for creating new hypotheses from image data and communicating them to researchers. \cite{horton2023large,manning2024automated} leverage LLMs to generate novel hypotheses and accelerate the research process. Similarly, we propose an LLM-based tool that forms hypotheses based on text (and statistically validates them on held-out documents).
\citet{vafa2022career} and \citet{du2024labor} demonstrate the value of purpose-built LLM tools for the analysis of labor-market data. Similarly, we aim to build a tool that is specifically aligned with solving a causal-inference task. Finally, we apply the econometric framework from \cite{llmapplied} to capture the use and limitations of LLMs across the causal-inference tasks in our application.

\section{Setup and Goal}
\label{sec:Setup}

We consider causal estimation based on text in a potential-outcomes framework with binary treatment $W \in \{1,0\}$. For each observation, we assume that the outcomes of interest are represented within a text document $Y \in \mathcal{Y}$, where documents may vary in length and we do not make any ex-ante restrictions on their structure. Following standard causal-inference notation \citep{neyman1923applications,rubin1974estimating}, we assume that each of these documents is equal to the potential document $Y(1)$ or $Y(0)$, depending on whether the corresponding observation is treated ($W=1$) or in control ($W=0$), so that $Y = Y(W)$. We aim to describe overall differences between the random variables $Y(1)$ and $Y(0)$.

We assume that we have access to data from a randomized trial where treatment $W$ is assigned randomly, that is, $(Y(1),Y(0)) \perp W$. As a consequence, describing differences between the distributions of $Y(1)$ and $Y(0)$ is equivalent to describing differences between the distributions of $Y|W{=}1$ and $Y|W{=}0$.
Our approach, therefore, can be seen as making statistical inferences on the difference between two distributions over text more broadly, where $W$ denotes group membership.
If these two groups indeed come from a randomized experiment, then we provide statistical inference on \emph{causal} effects; if the two groups denote any other two sources of texts (as is the case in our illustration in \autoref{sec:Empirical} below), then we still provide statistical inference on differences across the two distributions, even if the differences can no longer be interpreted causally.

To describe differences between the distributions $Y|W{=}1$ and $Y|W{=}0$ of text, we assume that we have a random sample available in which text from both distributions are sampled independently.

\begin{asm}[Random samples of text documents]
  \label{asm:Sampling}
  We have access to $n = n_1 + n_0$ observations $\{(Y_i,W_i)\}_{i=1}^n$ of $n_1 > 0$ treated texts $\{Y_i\}_{W_i = 1}$ and $n_0 > 0$ control texts $\{Y_i\}_{W_i = 0}$, where the two samples are independent of each other and each is sampled iid from the respective population distributions
  $Y_i \overset{d}{=} Y|W{=}1$
  for $W_i=1$
  and
  $Y_i \overset{d}{=} Y|W{=}0$
  for $W_i=0$. 
\end{asm}

Throughout, we perform our analysis conditional on the sizes $n_1$ and $n_0$ of the treated and control samples. This covers randomized trials with a fixed fraction of treated units as well as trials with a fixed probability of assignment. Our results are easily extended to popular alternative randomization schemes, such as stratified and cluster randomization.

Relative to a standard causal-inference framework in which realizations of the outcome $Y$ take the form of a single number or a low-dimensional vector, we face a fundamental challenge: if outcomes are represented in text form (or, more generally, are very complex), then we cannot simply summarize the effect of treatment $W$ on outcomes $Y$ in terms of simple estimands like differences $\E[Y|W{=}1] - E[Y|W{=}0]$ (corresponding to the average treatment effect $\E[Y(1) - Y(0)]$). Instead, we need to find ways of describing core aspects of the differences across $Y|W{=}1$ and $Y|W{=}0$ in a manageable and statistically valid way, based on a potentially small sample $\{(Y_i,W_i)\}_{i=1}^n$. In the following three chapters, we tackle this challenge by providing an approach to addressing \emph{whether} there is a difference between $Y|W{=}1$ and $Y|W{=}0$ (\autoref{sec:Testing}), \emph{what} the difference is (\autoref{sec:Describing}), and \emph{how complete} our description of said difference is (\autoref{sec:Completeness}).

\section{Testing for an Effect on Text}
\label{sec:Testing}

Our first goal is to detect whether the two distributions $Y|W{=}1$ and $Y|W{=}0$ over text documents are the same.
That is, before asking \emph{how} the two distributions differ, we first want to check \emph{whether} there is any statistically significant difference at all.
Even for this seemingly more effortless task, standard approaches for testing for differences across two distributions are limited.
First, we could choose a few features of the text documents by hand that we think could be affected by treatment and then run a joint test for average differences in those features across the two groups. For example, we could hypothesize that the texts are longer in one group than the other, that a specific word (or group of words) is more likely to appear in one group, and/or that the sentiment of the text (as measured by some pre-defined algorithm) differs in a systematic way. But the number of ways in which we could provide simple features of complex texts is so vast that this approach is likely to miss important signals unless we have strong and correct priors about which aspect of text is most affected.
Second, we could instead search through \emph{many} such features automatically by trying them one by one. However, this approach risks leading to spurious findings, requiring us to apply severe corrections to avoid excess false positives. In the context of a very high number of potential features relative to a small or moderate number of text documents, these corrections may make detecting significant effects very hard, if not impossible.
Finally, we could apply general non-parametric tests that have been proposed for detecting the difference between two distributions, such as multi-dimensional generalizations of the Kolmogorov--Smirnov test. Yet even such tests would rely on specific representations of the text, as they typically focus on relatively low-dimensional outcome vectors and may quickly lose power otherwise.

Instead of relying on tests that compare text features across the two groups, we adapt a trick from \cite{Ludwig2017-xk} and ask: Given a document $Y_i$ in our sample, can we predict its group assignment $W_i$ better than some trivial benchmark? That is, rather than asking whether $W$ has a causal effect on $Y$, we ask whether $Y$ is predictive of $W$. This change in framing allows us to leverage non-parametric tools from machine learning directly, since these tools excel at predicting some simple outcome (here: $W_i$) from complex covariates (here: $Y_i$) by systematically searching over good predictors.
Of course, we now have to adapt this idea for the specific case of a small sample of complex text documents.
First, not all machine learning methods work well with text.
Second, text methods specifically may be very complex and do not come with any statistical guarantees that allow us to assess predictive signal without additional tests.
We, therefore, adapt this idea for text data and provide a specific testing scheme that ensures validity without making restrictive assumptions on the machine-learning method we use, remains computationally manageable, and still works well with small sample sizes.

We construct a two-step procedure to test whether text documents can predict group assignments (and thus, whether the two distributions are different). In the first step, we let an algorithm learn the relationship of text documents to group assignments. In the second step, we validate whether the learned relationship outperforms a trivial benchmark. To ensure that our inferences are statistically valid, we separate these two steps by sample splitting, following \cite{egami2022make}.
Specifically, we split the full sample $\{1,\ldots,n\}$ randomly into a training sample $\Tr$ and a hold-out sample $\HO$. Throughout, all of our results are conditional on the split into training and hold-out.

\begin{asm}[Random training and hold-out samples]
  \label{asm:Splitting}
  The hold-out sample $\HO$ of size $h = h_1 + h_0$ consists of a subset of $\{i; W_i {=} 1\}$ of size $h_1$ (with $0 {<} h_1 {<} n_1$)
  and a subset of $\{i; W_i {=} 0\}$ of size $h_0$ (with $0 {<} h_0 {<} n_0$), both drawn uniformly at random and independently of each other.
  The training sample $\Tr$ is its complement, $\Tr = \{1,\ldots,n\} \setminus \HO$.
\end{asm}

We first apply a machine-learning algorithm to the training data $(Y_i,W_i)_{i \in \Tr}$ and have it produce predictions $\widehat{W}_i$ for the (unseen) group assignments $W_i$ of each text $Y_i$ in the hold-out, $i \in \HO$.
These predictions can be classifications ($\widehat{W}_i \in \{1,0\}$) or probabilities ($\widehat{W}_i \in [0,1]$).
On the hold-out data, we then estimate how well these predictions $\widehat{W}_i$ line up with the actual group assignments $W_i$ by calculating a loss $L((W_i,\widehat{W}_i)_{i \in \HO})$.
For classifications $\widehat{W}_i \in \{1,0\}$, this loss could be the misclassification rate or the F1 score.
For probabilities $\widehat{W}_i \in [0,1]$, it could be the mean-squared error, log-likelihood, or the area under the curve (AUC).
In any case, we can gauge the quality of these predictions by comparing them to the loss $L((W_i,\widehat{w})_{i \in \HO})$ for a trivial benchmark $\widehat{w} \in \{0,1\}$ or $\widehat{w} \in [0,1]$ obtained from the training data (for example, by minimizing $L((W_i,\widehat{w})_{i \in \Tr})$) that does \emph{not} depend on the text data.
Intuitively, if we obtain lower loss from using the text documents, $L((W_i,\widehat{W}_i)_{i \in \HO}) < L((W_i,\widehat{w})_{i \in \HO})$, then we have found evidence that there is signal for predicting group assignments -- and thus that the distributions over the two groups must differ.

As our preferred implementation, we use an LLM to form classifications $\widehat{W}_i \in \{0,1\}$, and then evaluate prediction quality using the accuracy or, for very imbalanced groups, the F1 score.%
\footnote{
\label{ftn:accuracyf1}
    The accuracy is the fraction $\text{Acc}((W_i,\widehat{W}_i)_{i \in \HO}) = \sfrac{\sum_{i \in \HO} \Ind{\widehat{W}_i = W_i}}{h}$ of correctly labeled instances.
    The F1 score is the harmonic mean $\text{F1}(\cdot) = 2 
    \frac{\text{Pre}(\cdot) \: \text{Rec}(\cdot)}{\text{Pre}(\cdot) + \text{Rec}(\cdot)}$
    of precision $\text{Pre}((W_i,\widehat{W}_i)_{i \in \HO}) = \sfrac{\sum_{\widehat{W}_i = 1} W_i}{\sum_{\widehat{W}_i = 1} 1}$ and recall $\text{Rec}((W_i,\widehat{W}_i)_{i \in \HO}) = \sfrac{\sum_{W_i = 1} \widehat{W}_i}{\sum_{W_i = 1} 1}$.
    In these two cases, we choose $L((W_i,\widehat{W}_i)_{i \in \HO}) = - \text{Acc}((W_i,\widehat{W}_i)_{i \in \HO})$ and $L((W_i,\widehat{W}_i)_{i \in \HO}) = - \text{F1}((W_i,\widehat{W}_i)_{i \in \HO})$, respectively.
    The F1 score has the downside that it is less intuitive and not invariant to switching group labels, but it may be the better choice when group labels are very imbalanced, in which case the standard practice is to designate the smaller group the positive class ($W_i = 1$).
}
In principle, the above procedure works with any machine-learning tool that can extrapolate the relationship of text documents $Y_i$ to assignments $W_i$ from the training to the hold-out samples.
In practice, the procedure is particularly easy to implement by prompting a state-of-the-art LLM with a large context window.\footnote{In our application in \autoref{sec:Empirical}, we use Google Gemini 1.5 Pro, which allows for a context window of around one million tokens or around 250,000 words, with implementation details and prompts provided in \autoref{apx:Implementation}.}
Specifically, we provide all training documents with their group assignment along with all held-out documents without their assignment in a single prompt and query the LLM to provide the most likely group assignments for each of the held-out text documents.
Beyond the simplicity of its implementation, one advantage of this LLM-based approach is that it is effectively pre-trained on a large corpus and on related text-analysis tasks, making this approach effective even with a small number of documents.
At the same time, prompting the LLM in this way only produces binary classifications (rather than group probabilities) and does not directly optimize for any specific loss function.
An alternative approach that still leverages LLMs would be to obtain low-dimensional embeddings from an LLM, and then obtain more granular probability estimates via non-parametric prediction.

Next, we construct a valid test of a difference in distributions from the chosen measure of out-of-sample prediction quality.
Rather than relying on assumptions on the machine-learning tool or large-sample arguments, we use a permutation test to turn prediction quality into a $p$-value that is valid irrespective of method and sample size, and instead relies on random sampling and sample splitting alone.
Specifically, we consider random permutations $\pi$ of the hold-out units $\HO$, and then compare the improvement in performance
\[
  \Delta = L((W_i,\widehat{w})_{i \in \HO}) - L((W_i,\widehat{W}_i)_{i \in \HO})
\]
over trivial predictions to the distribution of improvements
\(
  \Delta_\pi = L((W_{\pi(i)},\widehat{w})_{i \in \HO}) - L((W_{\pi(i)},\widehat{W}_i)_{i \in \HO})
\)
where actual treatment assignments are randomly permuted.
Intuitively, if there is no true difference between the two groups, then the loss improvement $\Delta$ on the actual data should not look systematically different from the improvement $\Delta_\pi$ when group assignments are shuffled. If, on the other hand, the predictions $\widehat{W}_i$ contain true signal about the assignments $W_i$, then we hope that the improvement $\Delta$ in loss is larger than many of the other draws $\Delta_\pi$.
This suggests the $p$-value
\begin{align*}
  \widehat{p} = \frac{1 + \sum_{b=1}^B \Ind{ \Delta \leq \Delta_{\pi_b}}}{1 + B}
\end{align*}
with $B$ permutations $\pi_b$ of $\HO$ drawn independently and uniformly at random,
and the associated test with significance level $\alpha$ that rejects the null hypothesis $Y|W{=}1 \stackrel{d}{=} Y|W{=}0$ for $\widehat{p} \leq \alpha$.%
\footnote{
  The $p$-value measures the fraction of draws that lead to an improvement that is at least the same as $\Delta$, including the original data. Therefore, we add one to the numerator and denominator. As a result of this adjustment, the test is conservative for small $B$, and has asymptotically exact size as $B \rightarrow \infty$.
}
 As long as the algorithm that produces the treatment assignments $(\widehat{W}_i)_{i \in \HO}$ never has access to the held-out actual assignments $(W_i)_{i \in \HO}$, the resulting test provides valid size control.

\begin{prop}[Valid permutation-based test]
  \label{prop:Test}
  If $(\widehat{W}_i)_{i \in \HO} \perp (W_i)_{i \in \HO} \: | \: (Y_i,W_i)_{i \in \Tr},(Y_i)_{i \in \HO}$,
  then the test based on the permutation $p$-value $\widehat{p}$ provides (conditionally) valid size control in the sense that $\P(\widehat{p} \leq \alpha|(Y_i,W_i)_{i \in \Tr},(Y_i)_{i \in \HO}) \leq \alpha$ under the null hypothesis $Y|W{=}1 \stackrel{d}{=} Y|W{=}0$.
\end{prop}

This test has three advantages in our setting. First, it is valid without further assumptions on the algorithm that produces the predictions. This includes LLMs, provided we can guarantee that the held-out assignments $(W_i)_{i \in \HO}$ never entered the training process (which includes the LLM training data as well as any fine-tuning exercise). Second, the test does not require the repeated computation of predictions, so it is computationally efficient even when we use very complex algorithms to obtain predictions $(\widehat{W}_i)_{i \in \HO}$. Third, the test is valid even when the number of documents is small.
At the same time, the test remains inherently limited since it does not provide a clear indication of what the differences across groups are -- it only provides evidence that some differences exist.

\section{Describing Effects on Text}
\label{sec:Describing}

Above, we have proposed an approach to testing \emph{whether} two distributions over texts are systematically different. We now tackle the challenge of describing \emph{what} the systematic differences across the distributions $Y|W{=}1$ and $Y|W{=}0$ are.

We note that the task of describing systematic differences is challenging for at least two reasons. First, the space of potential differences is amorphous and huge, and even (or especially) a human analyst would struggle searching through it systematically. Second, unlike the approach in the previous section that relied on a statistical test of differences across texts, making substantive claims about which outcomes are affected relies on a definition of economic outcomes that is hard to formalize. While the former challenge points to the massive potential of using machine-learning tools that search through complex data and hypotheses, the latter challenge also foreshadows their inherent limitations.

We first define what we mean by describing differences across text distributions, leading to the notion of causal themes (\autoref{subsec:Describing-Themes}). We then propose an LLM-based implementation to hypothesize such themes (\autoref{subsec:Describing-Implementation}), before discussing their validation (Sections~\ref{subsec:Describing-Validation} and \ref{subsec:Describing-Combining}). We close this section by emphasizing the role of human analysts beyond scoring (\autoref{subsec:Describing-Intervention}).

\subsection{Causal Themes}
\label{subsec:Describing-Themes}

To tackle the challenge of describing systematic differences in complex outcomes across different groups, we consider a simple description vector $Y^f = f(Y) \in \mathcal{D}$ that includes core aspects of documents across the two groups.
Specifically, we score texts according to $k$ core themes, where $k$ is a small number. Each of these themes may correspond to a specific topic covered in a document, to its style, or to the sentiment it expresses. For each topic $j$, we obtain a score $Y_j^f \in \mathcal{D}_j$ within a simple set $\mathcal{D}_j$. That simple set will typically be a continuous or discrete scale, such as when the prevalence of a topic is scored from 0 to 100 or sentiment as $\{-1, 0, +1\}$.
Or it can be a discrete classification into a few groups, for example whether the poetic form of a document is a sonnet, a haiku, or neither ($\mathcal{D}_j = \{ \textnormal{sonnet}, \textnormal{haiku}, \textnormal{other} \}$).
In any case, we assume that descriptions are easily understood and that we can use standard statistical tools to estimate causal effects on the low-dimensional vector $(Y^f_1,\cdots,Y^f_k) = Y^f = f(Y) \in \mathcal{D} = \mathcal{D}_1 \times \cdots \times \mathcal{D}_k$.
By using these descriptions, we turn the task of considering the intractable effect of $W$ on full documents $Y$ into describing differences between scores $Y^f_j |W =1$ and $Y^f_j |W =0$.

An optimal scoring $Y^f = f(Y)$ that summarizes a complex document $Y$ in terms of core themes should provide mutually exclusive, differentially complete, and causally meaningful descriptions of the text. By \emph{mutually exclusive}, we mean that components $Y_j^f$ describe themes that do not overlap in meaning and are entirely separate.
By \emph{differentially complete}, we specifically refer to the ability of the given description to capture the systematic differences across the two groups as well as possible. (In this sense, a good set of \emph{causal} themes differs markedly from a complete description of the text, which instead may include features that are \emph{common} across groups.) Finally, by \emph{causally meaningful}, we mean that the aspects of text captured by the description are descriptions of plausible outcomes described in the text. This latter criterion is meant to rule out cases where descriptions do not have a clear relationship to the content of the text or represent aspects that are unhelpful in describing effects of interest.
These criteria represent an ideal that may be infeasible in practice since simple scorings of complex documents may inherently lose relevant information and it may not be feasible to formalize what makes a description accurate or meaningful.

\subsection{LLM-Based Implementation}
\label{subsec:Describing-Implementation}

How can we construct simple summaries of documents that represent systematic differences between treatment and control groups?
One approach would be to rely on pre-specified scoring functions applied to the documents.
For example, we could classify documents based on the counts of specific words, use a generic topic model trained on some other text data, or apply algorithms for sentiment analysis.
However, the use of pre-specified classifications is limited in at least three ways.
First, such scores may not be a good fit for the specific documents in our experiment.
Second, even if they are good at classifying documents overall, they may not represent the systematic differences across groups well.
Finally, pre-specified ways of classifying our documents may fall short of our goal of expressing flexibly what the (possibly unexpected) differences across documents are, since they may reflect our preconceptions more than the data.

Instead of relying on pre-specified classifications of documents, we learn plausible causal themes from the data itself, similar to \cite{pmlr-v162-zhong22a,zhong2023goaldrivendiscoverydistributional}.
That is, rather than pre-specifying themes and the associated scoring function $f$, which maps text $Y$ to descriptions $Y^f$, we learn themes directly from the training data $(Y_i,W_i)_{i \in \Tr}$. We do not take the stance that there is a unique, true set of themes that we aim to learn.  Instead, we hope to obtain \emph{some} suitable description of documents, and will later provide valid statistical inference based on scoring additional documents $(Y_i,W_i)_{i \in \HO}$ in the hold-out sample. In fact, having flexibility with respect to the technology that proposes themes is a feature of our approach, and our inference framework still delivers robust causal statements regardless of which representation is used.

For our implementation, we feed into a large-language model (LLM) the documents in the training data, along with their classification into treatment and control groups.
From this LLM, we obtain themes that represent systematic \emph{differences} between treatment and control documents in terms of core themes that differentiate them (similar to a D5 system, \citealp{zhong2023goaldrivendiscoverydistributional}, and earlier work in \citealp{pmlr-v162-zhong22a}).
Practically, we can exploit very large context windows of modern LLMs (e.g.\ more than one million tokens in Google Gemini 1.5 Pro) to obtain such themes by prompting, where we request the LLM to provide descriptions of systematic differences across groups and to provide a list of themes and corresponding scales, along with examples. Details on our specific settings for the LLM and prompts are provided in \autoref{apx:Implementation}, and \autoref{sec:Empirical} provides a worked-out empirical example.

While a prompting-based approach is straightforward to implement using publicly available LLMs,
there are natural extensions that are likely to further align the model with the task of finding causal themes at the cost of additional computational and/or human effort.
First, providing more detailed examples of successful scorings according to causal themes has the potential to improve the quality and reliability of output. For example, \cite{Zhou2023LIMALI} shows that feeding a language model a curated set of high-quality prompts and responses, even in a limited quantity, can significantly enhance its performance by aligning it more closely with the specific task requirements.
As an additional limitation, next-token prediction is at best a crude proxy for aligning the LLM with being able to find classifications that distinguish well between treatment and control groups, and adding explicit rewards to the optimization for finding themes that distinguish well between groups may further improve the LLM's ability to find differentially complete summaries.
An important criterion for successful causal themes is not only their ability to distinguish between groups statistically but also whether these themes correspond to meaningful economic outcomes.
Since that question is ultimately one of communication to human researchers, further aligning causal theme classifications based on reinforcement learning from human feedback (RLHF) could be helpful in ensuring that the LLM provides effective descriptions of meaningful outcomes.

\subsection{Alternative Approaches Based on Topic Models}
\label{subsec:Describing-Alternatives}

Our approach can be related to supervised topic modeling, which provides an alternative implementation. Topic models like those based on Latent Dirichlet Allocation \citep[LDA;][]{blei2003latent} estimate common topics across documents, where each topic corresponds to a distribution over words and each document is assumed to come from a mixture of topics.
Since we are interested in those topics that are related to differences across groups, we relate more directly to \emph{supervised} topic models \citep{mcauliffe2007supervised} that implement a version of LDA in which topics are also optimized for predicting an outcome of interest (here, the assignment to groups).
A main downside of (supervised or unsupervised) topic models is that their performance may be unreliable at low sample sizes.
We document such behavior in our empirical evaluation in \autoref{sec:Empirical} below.
Relative to such an implementation based on supervised topic models, leveraging readily available LLMs simplifies our technical implementation and still provides reliable results when the number of documents is low.
The LLM approach may also capture themes beyond simple distributions over words, which allows us to integrate aspects of sentiment or more complex distributions over text. At the same time, this more vaguely defined concept of a ``theme'' may be less mathematically tractable than the narrower notion of a ``topic'' in LDA models.

\subsection{Human-Based Validation}
\label{subsec:Describing-Validation}

Having obtained themes from the training data, our goal is now to obtain valid statistical inference on the differences in scores across the two groups in the held-out data. In doing so, we face a challenge: for the themes obtained from the training data, we need to obtain scores for the held-out documents. Once we have these scores, we can then obtain inference on their differences across treatment and control using standard statistical tools (like the permutation test from \autoref{sec:Testing}), where our sample-splitting procedure \citep[which follows][]{egami2022make} ensures that our inference remains statistically valid.

A first approach to obtaining score labels would be to use a machine-learning tool to provide scores for the held-out text documents $Y_i$, while withholding their group assignments $W_i$. In our implementation, this could be achieved by prompting the LLM to provide scores $\widehat{Y}^{\hat{f}}_{ij}$ for each held-out document $i \in \HO$ and each theme $j \in \{1,\ldots,k\}$. These score labels $\widehat{Y}^{\hat{f}}_{ij}$ then estimate the true scores $Y^{\hat{f}}_{ij} = \hat{f}_{j}(Y_i)$ corresponding to the theme $\hat{f}_j$ suggested by the LLM based on training data. For example, the LLM may suggest ``text talks about cars'' as a theme $\hat{f}_j$ with a classification into $\widehat{\mathcal{D}}_j = \{\text{yes},\text{no}\}$. The classification $\widehat{Y}^{\hat{f}}_{ij}$ then corresponds to the LLM's claim whether a specific text $Y_i$ talks about cars or not.

Is sample splitting enough to make valid inference on differences in theme scores across text distributions based on LLM-generated labels? In general, the answer is clearly \emph{no}. While sample splitting allows us to provide statistical inferences for differences in machine scores $\widehat{Y}^{\hat{f}}_j$ across the two groups, it does not provide any guarantee that these inferences are valid for the true scores $Y^{\hat{f}}_j$. For example, an LLM may label any text that includes a vehicle as a ``text that talks about cars'', even if that vehicle is a boat rather than a car -- a phenomenon we call \emph{theme score bias}.
In this case, we could obtain \emph{statistically} valid inference that there must be a difference across the two groups, even if that difference is about boats rather than cars. But crucially, we do not obtain \emph{substantively} valid inference for the claim that the difference is related to cars.
This problem persists even if we know that the LLM can score text well on average based on some training distribution, since average scoring performance does not guarantee valid inference conditional on the chosen theme ${\hat{f}}_j$. In other words, the kind of guarantee we would need for being able to make substantively correct (causal) inferences based on LLM-generated scores is at odds with those typically available for such machine-learning tools, and would at least require \emph{some} true scores $Y^{\hat{f}}_{ij}$ for verification (an idea we expand on in \autoref{subsec:Describing-Combining} below). In the parlance of \cite{llmapplied}, we are applying an LLM to an estimation problem for which machine-generated scores are not sufficient to obtain valid inference.

Instead of relying on machine-learning scores for the held-out documents, we therefore take the stance that providing scores for the final validation should ultimately be left to human experts. That is, we assume that the themes provided by the procedure in \autoref{subsec:Describing-Implementation} can be scored correctly by experts, yielding true scores $Y^{\hat{f}}_{ij}$. Of course, such scoring may be costly, but we argue that it cannot be avoided if our goal is to make valid inference on these themes. Furthermore, the cost appears small relative to two counterfactual procedures that do not leverage machine learning: first, a procedure in which we generate costly scores for a small number of pre-specified themes, risking that we waste costly scores on aspects of the text that may only have a small chance of being affected; or second, a procedure that generates costly scores for a very large number of themes first and then inspects effects on all of these themes. In the first case, we spend a similar amount of effort on generating scores, but risk achieving much smaller power; in the second case, we spend a considerably larger cost on scoring without a clear gain in yield.

If we have true (human-generated) scores for the proposed (machine-generated) themes available, and none of the hold-out data entered the construction of the themes, then inference on differences across the two groups is straightforward. For example, as long as the hold-out sample is sufficiently large, we can obtain conditionally valid inference based on a central limit theorem for the score differences across the two groups.

\begin{prop}[Valid inference on themes]
  \label{prop:Inference}
  Assume that all the scores $Y^{\hat{f}}_j \in \R^k$ are uniformly bounded,
  that $\Var(Y^{\hat{f}}|\hat{f})$ is a.s.\ uniformly bounded away from zero,%
  \footnote{
    For the general multi-dimensional case, by ``uniformly bounded away from zero'' we refer to the smallest eingenvalue of $\Var(Y^{\hat{f}}|\hat{f})$, that is, $\Var(Y^{\hat{f}}|\hat{f}) \succeq \I \varepsilon$ for some universal $\varepsilon > 0$.
  }
  and that $\hat{f} \perp (Y_i,W_i)_{i \in \HO} \: | \: (Y_i,W_i)_{i \in \Tr}$.
  Define
  \(
    \tau = \E[Y^{\hat{f}}|W{=}1,\hat{f}] - \E[Y^{\hat{f}}|W{=}0,\hat{f}]
  \)
  and consider the estimator
  \[
    \widehat{\tau} = 
    \frac{1}{h_1} \sum_{i \in \HO; W_i = 1} Y^{\hat{f}}_i - \frac{1}{h_0} \sum_{i \in \HO; W_i = 0} Y^{\hat{f}}_i.
  \]
  Then 
  \begin{align*}
    \E[\widehat{\tau}|\hat{f}] &=  \tau
    \text{ a.s.}
    &
    &\text{and}
    &
    \sup_{t \in \R^k}
    \left|
    \P\left(
      \sqrt{\underline{h}} (\widehat{\tau} - \tau)
      \leq t
    \middle|
      \hat{f}
    \right)
    -
    \P\left(
      \N(\0_k,\Sigma)
      \leq
      t
    \middle|
      \hat{f}
    \right)
    \right|
    &\stackrel{p}{\longrightarrow}
    0
  \end{align*}
  as $\underline{h} = \min(h_1,h_0) \rightarrow \infty$,
  where
  $\Sigma = \sfrac{\underline{h}}{h_1} \Var(\hat{f}(Y)|W{=}1,\hat{f}) + \sfrac{\underline{h}}{h_0} \Var(\hat{f}(Y)|W{=}0,\hat{f})$.
\end{prop}

This result justifies constructing standard-error estimates, confidence intervals, and tests by estimating the variance as
\begin{align*}
  \widehat{\Var}(\widehat{\tau}|\hat{f}) = \frac{1}{h_1} \widehat{\Var}(\hat{f}(Y)|W{=}1,\hat{f}) + \frac{1}{h_0} \widehat{\Var}(\hat{f}(Y)|W{=}0,\hat{f})
\end{align*}
on the hold-out,
or by bootstrap inference based on resampling the hold-out sample only (holding $h_1,h_0$ fixed). In both cases, our conditional-inference approach means that we do not have to re-derive themes themselves, so no additional computation is necessary beyond straightforward calculations on the hold-out. The resulting inference is valid conditional on the scoring function $\hat{f}$ and does not require any guarantees about how topics were constructed. Instead, this approach only relies on sampling and sample splitting as a guarantee (Assumptions~\ref{asm:Sampling} and \ref{asm:Splitting}).

\subsection{Combining Cheap and Costly Scores}
\label{subsec:Describing-Combining}

Our method analyzes causal effects based on text data by summarizing text across treatment and control groups in terms of core themes.
While these summaries may be suggestive of causal effects,
we argue above that typical AI systems like the LLM we deploy in our empirical example in \autoref{sec:Empirical} do not come with guarantees about their validity, necessitating the use of human-generated scores.

In practice, scoring all hold-out documents may be costly. An alternative approach, going back to at least \cite{bound1991extent,pepe1992inference,robins1994estimation} and recently popularized in machine learning as ``prediction-powered inference'' (\citealp{angelopoulos2023prediction}; see \citealp{{ji2025predictions}} for a recent survey), is to combine costly (human-generated) true scores with cheap (machine-generated) approximations to correct for biases in the latter.
In our context, a natural way of combining true scores $Y^{\hat{f}}$ with machine-generated ones $\widehat{Y}^{\hat{f}}$ is the estimator
\begin{align}
  \label{eqn:Combined}
  \widehat{\tau}^\dagger
  =
  \frac{1}{h_1} \sum_{i \in \HO; W_i=1} \widehat{Y}^{\hat{f}}_i
  -
  \frac{1}{h_0} \sum_{i \in \HO; W_i=0} \widehat{Y}^{\hat{f}}_i
  -
  \left(\frac{1}{\ell_1} \sum_{i \in \Lb; W_i=1} (\widehat{Y}^{\hat{f}}_i {-} Y^{\hat{f}}_i)
  -
  \frac{1}{\ell_0} \sum_{i \in \Lb; W_i=0} (\widehat{Y}^{\hat{f}}_i {-} Y^{\hat{f}}_i)\right)
\end{align}
where $\Lb \subseteq \HO$ is a subset of the hold-out sample for which true scores $Y^{\hat{f}}_i$ are available.
Here, $\Lb$ is chosen uniformly at random among subset 
of $\HO$ of size $\ell = \ell_1 + \ell_0$ with $\ell_1>0$ units with $W_i{=}1$ $\ell_0>0$ units with $W_i{=}0$.
The estimator in \eqref{eqn:Combined} has an intuitive form, similar to the doubly-robust semi-parametric estimators from \cite{robins1994estimation} and \cite{chernozhukov_doubledebiased_2018}:
The first part can be thought of as a potentially biased estimator based on the machine-generated scores.
The second part estimates and corrects for the bias of the first part.
As a result, the full estimator is exactly unbiased (conditional on themes) for the true expected difference $\tau$ across groups.
This estimator is a special case of the bias-corrected estimator in \cite{angelopoulos2023prediction} and \cite{egami2024using}, and is also proposed by \cite{llmapplied} for the combination of human and LLM labels in the context of estimation problems.

If machine-generated scores closely resemble those by human experts, then the bias-correction term in \eqref{eqn:Combined} is small, and the estimator performs similarly to an estimator that uses true scores from the full sample. In this case, adding human-generated scores can improve precision. But the estimator is still unbiased and allows for valid inference even if the machine-generated scores are systematically off, without any substantive assumption on how the machine scores are generated.

\begin{prop}[Valid inference with combined scores]
  \label{prop:Combining}
  Assume that both true scores $Y_{ij}^{\hat{f}} $ and machine-estimated scores $\widehat{Y}_{ij}^{\hat{f}}$ are real-valued and uniformly bounded.
  If $Y_{ij}^{\hat{f}} = \hat{f}_j(Y_i)$ and $\widehat{Y}_{ij}^{\hat{f}} = \hat{f}^\dagger_j(Y_i)$ with
  $\hat{f}, \hat{f}^\dagger \perp (Y_i,W_i,L_i)_{i \in \HO} \: | \: (Y_i,W_i)_{i \in \Tr}$ and $\Var(Y^{\hat{f}}|W{=}1,\hat{f})$, $\Var(\widehat{Y}^{\hat{f}}-Y^{\hat{f}}|W{=}1,\hat{f},\hat{f}^\dagger)$ are both uniformly bounded from below., then 
  \begin{align*}
    \E[\widehat{\tau}^\dagger|\hat{f},\hat{f}^\dagger] &=  \tau
    \text{ a.s.}
    &
    &\text{and}
    &
    \sup_{t \in \R^k}
    \left|
    \P\left(
      \sqrt{\underline{\ell}} (\widehat{\tau}^\dagger - \tau)
      \leq t
    \middle|
      \hat{f}, \hat{f}^\dagger
    \right)
    -
    \P\left(
      \N(\0_k,\Sigma^\dagger)
      \leq
      t
    \middle|
      \hat{f}, \hat{f}^\dagger
    \right)
    \right|
    &\stackrel{p}{\longrightarrow}
    0
  \end{align*}
  as $\underline{\ell} = min(\ell_1,\ell_0) \rightarrow \infty$,
  where
  \begin{align*}
    \Sigma^\dagger &= \sfrac{\underline{\ell}}{\ell_1} 
    \left(\sfrac{\ell_1}{h_1}
  \Var(\hat{f}(Y)|W{=}1,\hat{f})
  +
  \left(1 - \sfrac{\ell_1}{h_1}\right) \Var(\hat{f}^\dagger(Y) - \hat{f}(Y)|W{=}1,\hat{f},\hat{f}^\dagger)\right)
  \\
  &\peq
  +
  \sfrac{\underline{\ell}}{\ell_0} 
  \left(\sfrac{\ell_0}{h_0}
  \Var(\hat{f}(Y)|W{=}0,\hat{f})
  +
  \left(1 - \sfrac{\ell_0}{h_0}\right) \Var(\hat{f}^\dagger(Y) - \hat{f}(Y)|W{=}0,\hat{f},\hat{f}^\dagger)\right).
  \end{align*}
\end{prop}

In particular, this result justifies estimating the variance of the estimator $\widehat{\tau}^\dagger$ by
\begin{align*}
  \widehat{\Var}(\widehat{\tau}^\dagger|\hat{f}) 
  &= \frac{1}{\ell_1} 
  \left(\sfrac{\ell_1}{h_1}
\widehat{\Var}(Y^{\hat{f}}|W{=}1,L{=}1,\hat{f})
+
\left(1 - \sfrac{\ell_1}{h_1}\right) \widehat{\Var}(\widehat{Y}^{\hat{f}} - Y^{\hat{f}}|W{=}1,L{=}1,\hat{f},\hat{f}^\dagger)\right)
\\
&\peq
+
\frac{1}{\ell_0} 
\left(\sfrac{\ell_0}{h_0}
\widehat{\Var}(Y^{\hat{f}}|W{=}0,L{=}1,\hat{f})
+
\left(1 - \sfrac{\ell_0}{h_0}\right) \widehat{\Var}(\widehat{Y}^{\hat{f}} - Y^{\hat{f}}|W{=}0,L{=}1,\hat{f},\hat{f}^\dagger)\right)
\end{align*}
on the hold-out (where $L_i = \Ind{i \in \Lb}$), or by a hold-out bootstrap (that holds $h_1,h_0,\ell_1,\ell_0$ fixed).
The variance expression has two components for each of the groups. The first part involves the variance of $Y^{\hat{f}} = \hat{f}(Y)$, which is the part of the variance that cannot be avoided since it stems from inherent within-group variation of the true scores.
The second part is about the difference $\widehat{Y}^{\hat{f}} - Y^{\hat{f}} = \hat{f}^\dagger(Y) - \hat{f}(Y)$ between estimated and true scores. This part is small if the estimated scores $\widehat{Y}^{\hat{f}}=\hat{f}^\dagger(Y)$ are a good approximation of the true scores $Y^{\hat{f}}=\hat{f}(Y)$. Its contribution goes down as the fraction $\sfrac{\ell_1}{h_1}$ or $\sfrac{\ell_0}{h_0}$, respectively, of true scores increases.

\subsection{Human Intervention before Validation}
\label{subsec:Describing-Intervention}

Above, we have considered a scheme by which an AI system proposes themes that may vary systematically across groups.
We have evaluated how these themes actually vary across groups in a validation step on the hold-out data.
However, in practice, not all proposed themes may provide the desired scientific insight, and it may be infeasible to prompt a large language model to achieve all the properties specified in \autoref{subsec:Describing-Themes}.
We, therefore, now consider an additional human validation step before scores are generated and effects estimated.

One reason why proposed themes may miss the mark is that they may not express relevant outcomes of interest. For example, if our procedure is applied to estimate causal effects in an experiment, then a theme proposed in the first stage may express differences related to the administration of treatment itself rather than plausible downstream differences in outcomes.
As an example, consider an application where treatment consists of attending a job-training program. Texts written by program participants may differ between treatment and control groups by whether they mention having participated in a job-training program, which represents a statistically correct difference across texts, but may not constitute a \emph{causally meaningful} one in the parlance of \autoref{subsec:Describing-Themes}.

Our procedure allows for human intervention to select and adjust themes before validation, provided that the ``firewall principle'' \citep{mullainathan2017machine} is not violated.
In general, different themes may be of different scientific value, which only researcher judgement may be able to assess. Our procedure remains valid if the researcher excludes, modifies, or adds themes after they are generated from the hold-out sample, but before validation. As long as the sample is split randomly between the training and hold-out samples and the researcher does not have access to any data from the hold-out sample when making these decisions, the downstream (conditional) inference on themes remains valid. Formally, these modifications become part of the themes $\hat{f}$ produced from the training data, only that they are now resulting from a joint human––AI procedure.

\section{Assessing the Completeness of Descriptions}
\label{sec:Completeness}

Given a set of themes that describe systematic differences across groups, we now ask how complete that description is.
That is, we aim to estimate which fraction of the differences in text is captured by the differences in these specific themes.
To do so, we go back to the prediction approach from \cite{Ludwig2017-xk} and ask how well the theme scores predict group membership.
Following the approach of \cite{fudenberg2022measuring}, we then benchmark this prediction performance against two alternatives: first, a trivial predictor that does not use any text information; and second, the non-parametric benchmark from \autoref{sec:Testing}.

\begin{defi}[Completness of descriptions]
  For themes $f:\mathcal{Y} \rightarrow \mathcal{D}$ and a loss function $L: (\{1,0\} \times \mathcal{W})^{\HO}$ defined over hold-out assignments and predictions as in \autoref{sec:Testing}, we define the completeness of the description provided by $f$ as
  \begin{align*}
      \textnormal{Completeness}(f)
      = \frac{
          \min_{w \in \mathcal{W}} \E[L((W_i, w)_{i \in \HO})]
          -
          \min_{\phi:  \mathcal{D} \rightarrow \mathcal{W}} \E[L((W_i, \phi(f(Y_i)))_{i \in \HO})]
      }{
        \min_{w \in \mathcal{W}} \E[L((W_i, w)_{i \in \HO})]
        -
        \min_{g:  \mathcal{Y} \rightarrow \mathcal{W}} \E[L((W_i, g(Y_i))_{i \in \HO})]
      },
  \end{align*}
  where we assume that there is some predictive information in the text for group labels so that the denominator is strictly positive.
\end{defi}

By construction, this completeness measure is always between 0 and 1. This is because using themes to predict group membership has at least as much (expected) predictive power as a constant prediction but restricts how documents can be used for predicting and thus has (expected) predictive power that is not better than the non-parametric benchmark in the denominator.
When applied to themes $\hat{f}$ learned on the training data, a completeness of close to 1 means that the information in those themes can separate groups as well as possible.
On the other hand, if the theme-based predictor performs considerably worse than the unconstrained benchmark, then the themes are limited in capturing systematic differences.

In practice, we can estimate completeness on the hold-out data by computing
\begin{align*}
  \widehat{\textnormal{Completeness}}(f)
  = \frac{
      L((W_i, \hat{w})_{i \in \HO})
      -
      L((W_i, \hat{\phi}(f(Y_i)))_{i \in \HO})
  }{
    L((W_i, \hat{w})_{i \in \HO})
      -
    L((W_i, \hat{g}(Y_i))_{i \in \HO})
  },
\end{align*}
where $\hat{w}, \hat{\phi}, \hat{g}$ are estimates of the respective minimizers that are learned based on the training data only.
As a result, the empirical measure can, in principle, be outside the unit interval.
This can happen because having more detailed information available to predict group assignment does not necessarily translate to better prediction performance in practice. For example, the complex non-parametric prediction $\widehat{W}_i = \hat{g}(Y_i)$ may overfit to the training sample and have worse out-of-sample performance than the simpler prediction $\widehat{W}^f_i = \hat{\phi}(Y^f_i)$ based on theme scores, assuming that the theme representation provide an effective low-dimensional representation of the text.

In the implementation we demonstrate in \autoref{sec:Empirical} below, we use the same LLM as in \autoref{sec:Testing} to form classifications $\widehat{W}_i = \hat{g}(Y_i) \in \{0,1\}$, obtain classifications $\widehat{W}^f_i = \hat{\phi}(Y^f_i)$ by logistic regression on the training sample, and evaluate prediction quality using accuracy.
We can apply this approach to the themes $f = \hat{f}$ proposed based on the training data, but also to other themes (such as those proposed by a human analyst) to compare their completeness.

\section{Empirical Proof of Concept}
\label{sec:Empirical}

We now evaluate the above approach by describing and evaluating systematic differences between two groups of documents.
Specifically, we consider two groups of academic abstracts. We use our approach to provide statistical evidence that the two groups are different; to describe systematic differences across the two; and to calculate how complete our description of these differences is.
We also show the return to mixing human and machine scores, and compare the main results to an alternative approach based on supervised topic models.
Our illustration demonstrates that our LLM-based implementation can provide clear evidence across the two groups of abstracts as well as easily understandable summaries, even when sample sizes are very low.

\subsection{Data}
\label{subsec:Empirical-Data}

In designing a plausible dataset to apply our method, we face the challenge that our underlying LLM may already have a large amount of publicly available text in its training data. We, therefore, create our own grouping of abstracts of academic manuscripts from scratch to avoid contamination and choose a labeling that does not obviously align with common categorizations. While the LLM may have seen the actual text in these abstracts before, it has not seen this (human-generated) classification into groups. Our sample-splitting approach, therefore, still provides valid statistical inference.

We start with 200 \textit{arXiv} abstracts that are randomly chosen from the econometrics category. A member of our research team then subjectively labels each of the abstracts by whether they align with their own research interests (``treatment'', Group A, $W=1$) or not (``control'', Group B, $W=0$). We employ this human-generated, subjective classification to see whether the causal LLM can capture structure that goes beyond easily measurable differences, such as differences across pre-defined categories or the usage of specific keywords.
We (randomly) divide these 200 abstracts into a first half of ``training'' abstracts and a second half of ``hold-out'' abstracts.

\subsection{Is There a Difference Between Corpora?}
\label{subsec:Empirical-Testing}

First, we check whether the two groups are indeed different. We implement the reverse-prediction approach from \autoref{sec:Testing}, where we obtain a non-parametric classification of documents via an LLM (Google Gemini 1.5 Pro via prompting as in \autoref{apx:Implementation}).
We measure performance via accuracy (the fraction of correctly labeled instances).
We also report F1 scores (the harmonic mean of precision and recall, see \autoref{ftn:accuracyf1}) as an alternative measure of performance that takes into account the imbalance in class assignment in our dataset (the treatment group is smaller than the control).
The LLM-based classification achieves an out-of-sample accuracy of 86\% (F1 score of .74), relative to a trivial classification that always guesses Group B (``control'') and achieves an out-of-sample accuracy of 71\% (F1 score of .00 since its recall is necessarily zero).
A permutation test on the hold-out shows that this difference is significant, with none of 10,000 random permutations achieving a better improvement in accuracy (or in F1 score) over the trivial benchmark ($p$-value of 0.000). 

\subsection{What Is the Difference Across Corpora?}
\label{subsec:Empirical-Describing}

Having established that the groups do not follow the same distribution, we next describe differences across groups.
As a starting point, we first prompt the LLM to describe systematic differences across treatment and control groups based on the training sample only, with results reported in \autoref{fig:summary}.
Following \autoref{sec:Describing} with the template in \autoref{apx:Implementation}, we then have the LLM propose up to six specific themes and scales according to which documents can be scored. These themes and scales are reported in \autoref{tbl:scales}, which is produced by the LLM.
(We again leverage Google Gemini 1.5 Pro via prompting, using a new instance relative to \autoref{subsec:Empirical-Testing}.)
We then have a different member of our research team score hold-out abstracts according to these scales.
In addition, we also obtain machine-generated scores for both training and held-out abstracts.
As a result, we can tabulate average machine-generated scores across treatment and control groups in both the training and hold-out samples, and on the latter, we can compare machine- to human-generated scores.

\begin{figure}[p]
  \begin{Verbatim}
Group A focuses on causal inference and econometric methods, often addressing issues like endogeneity, selection bias, and treatment effects. These documents employ a formal and technical writing style, emphasizing mathematical rigor and statistical inference. They often present novel methodologies or theoretical frameworks for addressing specific econometric challenges.

In contrast, Group B delves into a broader range of applied economics and econometrics topics, including financial markets, time series analysis, network analysis, and auction theory. These documents tend to be more applied, focusing on empirical analysis, modeling techniques, and specific economic phenomena.

The writing style in Group B is generally less technical than in Group A, although still employing statistical and mathematical language. While Group A emphasizes causal identification and estimation, Group B prioritizes model building, forecasting, and understanding economic relationships.

Furthermore, Group A often deals with microeconomic contexts, such as labor markets, education, and individual decision-making. Conversely, Group B encompasses both microeconomic and macroeconomic perspectives, exploring topics like inflation, GDP, and international trade.

Finally, Group A papers often propose new econometric methods or theoretical frameworks, while Group B papers tend to apply existing methods to new datasets or economic questions. This difference highlights the distinct objectives of the two groups: Group A aims to advance econometric methodology, while Group B focuses on utilizing these methods to gain insights into real-world economic issues.
  \end{Verbatim}
  \caption{Description of group differences provided by the LLM based on the training sample.}
  \label{fig:summary}
\end{figure}

\begin{table}
  \centering
  \begin{tabular}{llp{7cm}l}
    \toprule
    ID & Name & Description & Scale \\
    \midrule
    CEI &
    Causal Effects \& Identification &
    Focus on causal inference, identification strategies, and treatment effects (from not at all to strong emphasis) &
    0, 1, 2, 3 \\
    MET &
    Methodological Focus &
    Emphasis on developing new econometric methods or theoretical frameworks (from applied to highly methodological) &
    0, 1, 2, 3 \\
    MIC &
    Microeconomic Context &
    Focus on microeconomic agents, markets, and decision-making (from macroeconomic to microeconomic) &
    0, 1, 2, 3 \\
    FIN &
    Financial Markets \& Time Series &
    Emphasis on financial markets, time series analysis, and forecasting (from not at all to strong emphasis) &
    0, 1, 2, 3 \\
    NAP &
    Network \& Auction Applications &
    Focus on network analysis, game theory, and auction design (from not present to strong focus) &
    0, 1, 2, 3 \\
    \bottomrule
\end{tabular}

          \caption{Causal themes and scales provided by the LLM based on the training sample.}
    \label{tbl:scales}
\end{table}

The differences in average theme scores are provided in \autoref{tbl:Differences},  with conditional standard error estimates based on a bootstrap on the hold-out dataset only
Our main results are the differences in human-generated scores (\autoref{tbl:Differences-Human}), which show clear and statistically significant differences across groups for the first four of the five themes, while the last theme only shows small and noisy differences that could be consistent with equal means.
In addition, we also inspect machine-generated scores for both the hold-out (\autoref{tbl:Differences-Machine}) and the training abstracts (\autoref{tbl:Differences-Training}), with similar findings.
A joint Wald test of the null hypothesis that the score averages are the same across treatment and control group clearly rejects with $p$-value of close to zero, for both the human- and machine-generated hold-out scores.
As we would expect, the differences in the training sample (where treatment status was available during training) tend to be slightly larger than those in the held-out data (where treatment assignment was not available to the LLM).
This indicates slight overfitting and reinforces the need to use sample splitting to obtain valid causal inference.

\begin{table}
  \centering
\begin{subtable}{\textwidth}
  \centering
  \begin{tabular}{l S[table-format=1.2] S[table-format=1.2] S[table-format=1.2] S[table-format=1.2] S[table-format=1.2]}
    \toprule
    Group & {CEI} & {MET} & {MIC} & {FIN} & {NAP} \\
    \midrule
    A (``treatment'') & 1.45 & 2.10 & 0.55 & 0.34 & 0.34 \\
    & {(0.22)} & {(0.11)} & {(0.15)} & {(0.14)} & {(0.14)} \\
    B (``control'') & 0.24 & 1.72 & 0.89 & 0.73 & 0.38 \\
    & {(0.08)} & {(0.11)} & {(0.13)} & {(0.12)} & {(0.10)} \\
    \bottomrule
  \end{tabular}
  
  \caption{Average out-of-sample scores from human-scored data}
  \label{tbl:Differences-Human}
\end{subtable}

\bigskip

\begin{subtable}{\textwidth}
  \centering
  \begin{tabular}{l S[table-format=1.2] S[table-format=1.2] S[table-format=1.2] S[table-format=1.2] S[table-format=1.2]}
    \toprule
    Group & {CEI} & {MET} & {MIC} & {FIN} & {NAP} \\
    \midrule
    A (``treatment'') & 1.97 & 2.14 & 0.55 & 0.28 & 0.31 \\
    & {(0.23)} & {(0.11)} & {(0.19)} & {(0.16)} & {(0.16)} \\
    B (``control'') & 0.44 & 1.48 & 0.38 & 1.00 & 0.41 \\
    & {(0.11)} & {(0.07)} & {(0.11)} & {(0.15)} & {(0.10)} \\
    \bottomrule
  \end{tabular}
  
  \caption{Average out-of-sample scores from LLM-scored data}
  \label{tbl:Differences-Machine}

\end{subtable}

\bigskip

\begin{subtable}{\textwidth}
  \centering

  \begin{tabular}{l S[table-format=1.2] S[table-format=1.2] S[table-format=1.2] S[table-format=1.2] S[table-format=1.2]}
    \toprule
    Group & {CEI} & {MET} & {MIC} & {FIN} & {NAP} \\
    \midrule
    A (``treatment'') & 2.12 & 1.96 & 0.80 & 0.12 & 0.20 \\
    B (``control'') & 0.17 & 1.60 & 0.16 & 1.52 & 0.49 \\
    \bottomrule
  \end{tabular}

  \caption{Average in-sample scores from LLM-scored data}
  \label{tbl:Differences-Training}
\end{subtable}

  \caption{Average scores across groups and themes, with conditional standard error estimates based on a hold-out bootstrap.}
  \label{tbl:Differences}
\end{table}

\subsection{How Complete is the Description?}
\label{subsec:Empirical-Completeness}

To quantify how well the themes describe differences across treatment and control, we use the completeness metric from \autoref{sec:Completeness} with reverse-classification F1 scores as a metric.
As an upper benchmark, we use the reverse LLM-based classification from \autoref{subsec:Empirical-Testing} above, which we define as a completeness of 100\%.
The lower benchmark is obtained by trivial classification into the most common group, which yields an accuracy of 71\% and corresponds to a completeness of 0\%.
The results of the completeness exercise are reported in \autoref{tbl:Reverse}.

A logistic classification of treatment status from the human-labeled theme scores achieves an accuracy score of 85\%, where the logistic regression is learned on training data only.
This performance corresponds to a completeness score of 93\%, suggesting that the LLM-based themes capture differences exceptionally well.
Interestingly, a logistic regression based on machine-generated scores performs somewhat worse (accuracy of 81\%, completeness of 67\%) despite both the regressions being trained on LLM scores on the training data, consistent with an error-in-variables interpretation where LLM scores represent noisy versions of human scores.

In addition, we also evaluate a reverse classification based on a supervised topic model with five topics, as discussed in \autoref{subsec:Describing-Alternatives}.
The supervised topic model only achieves an out-of-sample accuracy of 73\% (completeness of 13\%), which shows the potential gain of scoring based on our causal LLM approach over summaries based on conventional topic models when sample sizes are very small.

We further test the hypothesis that each model performs equivalently to the trivial classifier, with corresponding $p$-values based on a permutation test for accuracy executed on the hold-out provided in \autoref{tbl:Reverse}. All the approaches described here document significant differences across groups, with the $p$-value for the supervised topic model at 4.1\% compared to 0.0\% for the LLM-based implementation.
These conclusions would be similar if we used F1 scores, where we note, however, that the best trivial classification would be to assign all instances to Group A (``treatment'') in this case for a benchmark F1 score of .29.

\begin{table}
  \centering

  \begin{tabular}{lrrrc}
      \toprule
      Method & Accuracy & F1 Score & Completeness & $p$-value\\
      \midrule
      Trivial classification & 71\% & 0.00 & 0\%  & --\\
      Direct LLM classification & 86\% & 0.74 & 100\%  & 0.000\\
      \midrule
      Logit based on LLM themes (human labels) & 85\% & 0.71 & 93\% & 0.000 \\
      Logit based on LLM themes (LLM labels) & 81\% & 0.67  & 67\% & 0.000\\
      Logit based on supervised topic model & 73\% & 0.27 & 13\% & 0.041\\
      \bottomrule
  \end{tabular}

    \caption{Hold-out evaluation in terms of reverse classification of treatment assignment, with $p$-values based on a hold-out permutation test of the null hypothesis that out-of-sample classification performance is no better than the trivial classification benchmark. The trivial classification maximizes accuracy, which here means assigning everybody to Group B (``control''). Completeness and $p$-values calculated based on accuracy.}
      \label{tbl:Reverse}
  \end{table}

\subsection{Combining Human and Machine Scores}
\label{subsec:Empirical-Combining}

We can further conceptualize the cost--benefit tradeoff of human scoring in terms of a tradeoff between the cost of producing score labels and the benefit of improved precision in our average treatment effect estimate at the theme level. Using the approach outlined in \autoref{subsec:Describing-Combining}, we can simulate this tradeoff by estimating $\hat{\tau}^\dagger_j$ for various sizes of the set $\Lb \subseteq \HO$ of human-scored instances in the hold-out set. In particular, for a given $\ell = |\Lb| \leq |\HO| = 100$, we estimate the variance of $\hat{\tau}^\dagger_j$ via bootstrap, and plot how that value changes with $\ell$ for each of our five identified themes (holding the fraction of ``treated'' fixed, $\sfrac{\ell_1}{\ell} \approx \sfrac{h_1}{h} = 29\%$). The resulting plot (\autoref{fig:labelcost}) shows that the precision of the estimator for each theme is monotonically decreasing in $\ell$. Furthermore, the curvature of each line illustrates decreasing marginal benefit to additional human scoring, implying that it may be optimal (from a decision-theoretic point of view) to human-label a certain number of held-out documents, with the stopping point characterized by the point at which the marginal benefit (in terms of estimator precision) outweighs the marginal cost (in terms of scoring cost).

\begin{figure}
    \centering
    \includegraphics[scale=0.5]{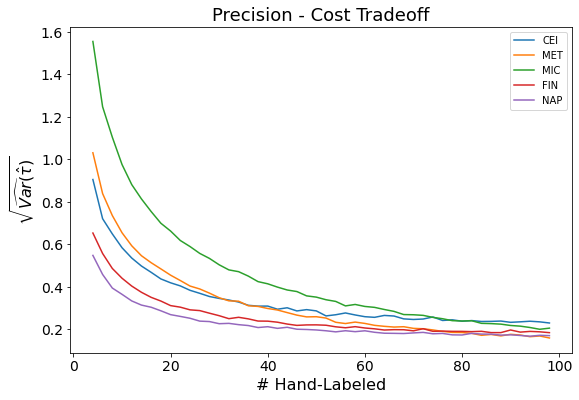}
    \caption{Tradeoff between the number of human-labeled hold-out datapoints ($\ell$) and the precision of the simple average treatment effect estimator for each theme.}
    \label{fig:labelcost}
\end{figure}

\section{Pre-Specification and Replicability}
\label{sec:Replication}

We propose a procedure that uncovers and validates systematic differences across complex text documents.
In theory, the procedure provides statistical validity even when we lack guarantees about the underlying AI tools.
However, when the procedure is used in practice, we may be concerned about our ability to replicate the procedure and verify its results.
In this section, we propose a flexible pre-specification approach based on sample splitting that balances the value of flexibility with the need for statistical validity and ex-post verification.

Concerns around replicability arise when the procedure is run multiple times and some estimation and reporting choices are adjusted based on the results, such as when prompts are changed, themes selected, or scores modified after looking at some of the hold-out evaluation.
With the use of large language models, there is an additional concern that the LLM explicitly or implicitly had access to data from the hold-out dataset.
A standard approach to mitigate such concerns is to rely on pre-specification, where all steps of the analysis are pre-registered in a pre-analysis plan before the data is accessed, and the procedure is then run under controlled conditions according to that pre-specification.

We believe that strict pre-analysis plans are limited in working for AI procedures like ours, both because they do not capture all concerns around replicability and because they can be too restrictive.
First, even when the analysis is pre-specified and all plans are followed, complex analysis based on LLMs (or other AI systems) may fail to replicate because it may produce different results when applied to the same data. This can be because the underlying LLM was changed by the developer or even because two runs of the \emph{same} model produce different outputs due to inherent randomness or instability (a challenge not usually overcome by setting randomness parameters such as the ``temperature'' to their minimum).
Second, strict pre-analysis plans can also be too restrictive because some corrections and adjustments based on a human inspection of preliminary results may be essential to producing meaningful output, such as expert judgment about which themes represent meaningful causal effects (\autoref{subsec:Describing-Intervention}).
We, therefore, see pre-specifying the entire algorithm before any data becomes available as a solution that is simultaneously insufficient and impractical.

Instead of pre-specifying the complete analysis, we propose a sequential-access analysis plan where the analyst first obtains access to the training data only, and a third party witholds access to the hold-out data (similar to split-sample analysis plans proposed in \citealp{anderson2017split}).
The analyst then runs all the preliminary analysis on the training data, adjusts prompts as necessary, and chooses themes.
Before the hold-out data becomes available, the analyst commits to the themes to be evaluated out of sample.
Only then is the remaining data made available on which the specified analysis is performed and reported, with the third party certifying its result.%
\footnote{
One challenge with sample splitting arises from our approach to testing based on LLMs  (\autoref{sec:Testing}).
Specifically, our proposed test for differences across groups requires access to texts from the hold-out data to predict group labels.
For the testing part, the pre-specification may therefore either include detailed instructions about how to implement scoring once the documents become available, possibly including prompts.
Or the procedure may involve an intermediate step where held-out documents are made available and predictions of group labels are registered before the remaining held-out data (that is, held-out group labels) are made available as well.
}
In our eyes, this approach balances the need for replicability with the benefit of adjustments, and does not rely on the full replicability of hypothesis generation.

\section{Conclusion and Discussion}
\label{sec:Conclusion}

We provide a tool for estimating causal effects on qualitative outcomes captured in text.
This method combines the power of modern AI tools with statistical guarantees.
In data from a randomized trial, it summarizes individual responses in terms of core themes and then provides causal inference on the resulting summaries.

\paragraph{Role of AI.}
Within this procedure, we view the role of generative AI as a hypothesis-generation machine.
This view does \emph{not} give AI a role in validating evidence by itself. Instead, it gives LLMs a role in generating plausible hypotheses that are then validated, similar to the approach of \cite{zhong2023goaldrivendiscoverydistributional}. For this procedure to work, hypotheses have to be formulated in a way that can be validated based on separate ground-truth data. In our framework, AI-generated themes play the role of these hypotheses, and we then use sample splitting and human scoring for validation.
We do not see this as a limitation of the current generation of AI solutions specifically but instead as a general structure for integrating machine learning and AI.

\paragraph{Researcher--AI complementarity.}
While the AI can suggest plausible causal effects, our validation and inference necessarily rely on human judgment: we cannot automatically evaluate whether causal themes are correct, meaningful, and relevant representations of causal effects expressed in text.
Beyond this need for validation, our approach also faces remaining challenges with respect to the interpretation of its output. 
On the one hand, an AI tool may be more systematic in picking up differences between treatment and control groups than a human analyst. At the same time, the analysis may still be driven by pre-conceptions embedded in the corpus the LLM is trained on and is thus not neutral.
Ultimately, we see the AI as a communication tool \citep[akin to][]{ludwig2024machine} that proposes summaries that can be evaluated statistically, but still have to be contextualized by an expert researcher to be evaluated substantively.

\bibliography{_references}

\newpage

\appendix

\section*{Appendix}

\section{Implementation Details}
\label{apx:Implementation}

Our implementation is based on Google Gemini Pro 1.5 (gemini-1.5-pro-001) with temperature set to 0.

\paragraph{Prompt for LLM providing reverse classification} (queried separately from prompts that describe differences, followed by training and hold-out documents formatted as JSON)

\begin{enumerate}
    \item[] 

    \begin{Verbatim}[frame=single, commandchars=\\\{\}]
Here is a list of documents, with associated group assignments. Note that the document IDs are randomly assigned. Your task is to learn to predict the group assignment (A or B) from the document text. You will be evaluated on another set of documents that I will pass to you. Do not give me any code or describe any other modeling mechanism: simply learn the relationship between text and group and wait for me to give you the testing set. Here is the training set:
\textbox{Training documents with treatment indicator}
Here is the testing set. Please report only the document ID and predicted group for each document. Please score every document.
\textbox{Testing documents without treatment indicator}
    \end{Verbatim}
\end{enumerate}

\paragraph{Prompts for LLM describing differences} (queried successively within the same context window, in a separate instance from the above)

    \begin{enumerate}
        \item Summary prompt (followed by training documents formatted as JSON)
        \begin{Verbatim}[frame=single, commandchars=\\\{\}]
Below is a list of documents that belong to groups A and B. Look at all the documents, and describe the main systematic differences in themes between documents in groups A and B, in about ten sentences. These differences could concern topics, tone and style of language, sentiment, or any other major differences across groups.
Here are the documents I want you to analyze, together with which group they are in. Note that the order is random, and the id is not informative about the topic or group.
\textbox{Training documents with treatment indicator}
        \end{Verbatim}
        \item Theme and score prompt
    \begin{Verbatim}[frame=single]
Provide up to six themes that help distinguish between groups A and B in the previously provided documents, based on their systematic differences. Each theme is one aspect of the text, such as a topic, sentiment, or characteristic. The themes should all be different (mutually exclusive), and also cover all the differences between groups A and B as well as possible (exhaustive with respect to differences). For each of the themes, create a scale to score each document. Specifically, these themes and scales could be of the following type:
(a) Topics. Score each document on a scale from 0 to 3 by how much a topic appears in the document.
(b) Language. Score a specific aspect of the style or tone on a scale from 0 to 3, such as how technical the language is or whether is uses a lot of metaphors.
(c) Sentiment. Score a specific aspect of the sentiment expressed by the text from -1 to +1, such as how positive or negative it deals with a specific topic.
If there are other types of themes or scales that are very fitting for the  application, please also use those.

Your output should take the same JSON form as in the following examples. First, your output should include a description of themes. Here is an example of themes that come from other data:

```
[
    {
        "theme_id": "TEC",
        "theme_name": "Technical language",
        "theme_description": "How technical the language is (from not technical to very technical)",
        "theme_scale": [0, 1, 2, 3]
    },
    {
        "theme_id": "APL",
        "theme_name": "Focus on applicability",
        "theme_description": "The degree to which the presented results can be applied in practice (from not at all to very well)",
        "theme_scale": [0, 1, 2, 3]
    },
    {
        "theme_id": "ANI",
        "theme_name": "Animals",
        "theme_description": "Emphasis on animals (from not at all to very much)",
        "theme_scale": [0, 1, 2, 3]
    },
    {
        "theme_id": "GSE",
        "theme_name": "General sentiment",
        "theme_description": "The general sentiment of the text (from negative to positive)",
        "theme_scale": [-1, 0, 1]
    },
    {
        "theme_id": "SRO",
        "theme_name": "Sentiment about robots",
        "theme_description": "The sentiment expressed about robots specifically  (from negative to positive)",
        "theme_scale": [-1, 0, 1]
    },
    ,
    {
        "theme_id": "POE",
        "theme_name": "Poetic form",
        "theme_description": "The poetic form of the document (actual form)",
        "theme_scale": ["Sonnet", "Haiku", "neither"]
    }
]
```

Second, a scoring of documents according to these themes:

```
[
    {
        ID1: TEC1,APL0,ANI3,GSE-1,SRO0,POEHaiku
    },
    {
        ID2: ...
    },
    ...
]
```

Note that "ID1" and "ID2" correspond to the document_id values of the above documents. The theme scores should be presented as they are above, with a comma-delimited list where each element comprises of the 3-letter theme_id followed by the score for the given document for that theme. Please score all of the documents.
    \end{Verbatim}

    \item Out-of-sample machine score prompt (followed by held-out documents formatted as JSON)
    
    \begin{Verbatim}[frame=single, commandchars=\\\{\}]
For each of the new documents below, for which I do not tell you whether they are in group A or B, please score these according to the themes as before. Your output should be of the same format as before. Here are the new documents:
\textbox{Testing documents without treatment indicator}
    \end{Verbatim}
\end{enumerate}

\newpage

\section{Details of Statistical Methods}

\paragraph{\autoref{tbl:Differences}.}

The standard errors reported in this table are produced via nonparametric bootstrap. For each of 10,000 bootstrap draws, we resample (with replacement) from our hold-out set. We then calculate the average theme scores for both groups of documents based in this resampled set, as well as the variance-covariance matrices for each group. The square roots of the diagonal of these matrices produce the standard errors contained in the table. To run the joint test for cross-group difference, we utilize a Wald test. Specifically, we assume that each mean vector is drawn from a Normal distribution:
\begin{align*}
    \tilde{\mu}_0 \sim \mathcal{N}(\mu_0, \Sigma_0) \qquad \tilde{\mu}_1 \sim \mathcal{N}(\mu_1, \Sigma_1)
\end{align*}
We then estimate $\widehat{\mu}_0, \widehat{\mu}_1, \widehat{\Sigma}_0, \widehat{\Sigma}_1$ using the bootstrap and construct our test statistic:
\begin{align*}
    Z = (\widehat{\mu}_1 - \widehat{\mu}_0)^\top(\widehat{\Sigma}_1 + \widehat{\Sigma}_0)^{-1}(\widehat{\mu}_1 - \widehat{\mu}_0)
\end{align*}
The $p$-value is then calculated via the corresponding quantile of the $\chi^2$ distribution.

\paragraph{\autoref{tbl:Reverse}.}

The $p$-values reported in this table are constructed via permutation tests. For each of 10,000 draws, we permute true group assignment labels within our held-out set. We then calculate accuracies for each model, comparing the permuted labels with the unadulterated model projected group assignments. We then construct a distribution of these accuracies. The $p$-value for each model listed in the table is then simply calculated as the mass of this distribution lying on or above the given model's accuracy based on unpermuted true labels. (We note that we do not have to normalize by the trivial classification accuracy for testing in this case since its accuracy is invariant to permutations.)

\paragraph{\autoref{fig:labelcost}.}

The variance of the simple average treatment effect estimators for a given number of human-labeled documents is calculated via nonparametric bootstrap. We produce 100,000 bootstrap values via a nested procedure. In the outer loop, we resample (with replacement) documents from our held-out set. Then, in the inner loop, we randomly assign each document to be in either the human-labeled set or not (i.e., we randomly decide whether to use a document's human-labeled theme scores or its LLM-labeled theme scores). We hold the fraction of human-labeled documents constant in each of the two groups. We then calculate the bias-corrected estimate of the average treatment effect. We run the outer loop 1000 times and the inner loop 100 times, generating 100,000 total bootstrapped values. We then simply estimate the variance across these draws for each theme.

\newpage

\section{Proofs}
\label{apx:Proofs}

\begin{proof}[Proof of \autoref{prop:Test}]
  Under the null hypothesis $Y|W{=}1 \stackrel{d}{=} Y|W{=}0$ (or equivalently, $Y \perp W$), we have that 
  $(\widehat{W}_i, W_{\pi(i)})_{i \in \HO} \: | \: (Y_i,W_i)_{i \in \Tr},(Y_i)_{i \in \HO}$
  is invariant to the choice of permutation $\pi$ of $\HO$.
  As a result, so is $\Delta_\pi | \: (Y_i,W_i)_{i \in \Tr},(Y_i)_{i \in \HO}$.
  For $B$ permutations $(\pi_b)_{b=0}^B$ of $\HO$ drawn uniformly at random,
  we thus have that
  \[
    (\Delta, \Delta_{\pi_1}, \ldots, \Delta_{\pi_B})
    \stackrel{d}{=}
    (\Delta_{\pi_0}, \Delta_{\pi_1 \circ \pi_0}, \ldots, \Delta_{\pi_B \circ \pi_0})
    \stackrel{d}{=}
    (\Delta_{\pi_0}, \Delta_{\pi_1}, \ldots, \Delta_{\pi_B})
    | \: (Y_i,W_i)_{i \in \Tr},(Y_i)_{i \in \HO},
  \]
  where we have also used that $(\pi_0,\pi_1 \circ \pi_0, \ldots, \pi_B \circ \pi_0) \stackrel{d}{=} (\pi_0,\pi_1, \ldots, \pi_B)$.
  In particular, $\Delta$ can be seen as drawn uniformly from the elements in $\Delta, \Delta_{\pi_1}, \ldots, \Delta_{\pi_B}$ conditional on $(Y_i,W_i)_{i \in \Tr},(Y_i)_{i \in \HO}$ (since the same holds for $\Delta_{\pi_0}$ from $\Delta_{\pi_0}, \Delta_{\pi_1}, \ldots, \Delta_{\pi_B}$ by exchangeability).
  Hence, writing $\Delta_{(1)}, \ldots, \Delta_{(B+1)}$ for a decreasing ordering of the $B+1$ elements in $\Delta, \Delta_{\pi_1}, \ldots, \Delta_{\pi_B}$,
  we have that $\P\left( \Delta > \Delta_{(b+1)} \right) \leq \frac{b}{1 + B}$ for all $b \in \{1,\ldots,B\}$.
  Then, for every $\alpha < 1$,
  \begin{align*}
    &\P\left(\widehat{p} \leq \alpha \middle| (Y_i,W_i)_{i \in \Tr},(Y_i)_{i \in \HO} \right)
    = 
    \P\left(1 + \sum_{b=1}^B \Ind{ \Delta \leq \Delta_{\pi_b}} \leq (1 + B) \alpha \middle| (Y_i,W_i)_{i \in \Tr},(Y_i)_{i \in \HO} \right)
    \\
    &=
    \P\left(\Delta > \Delta_{(\lfloor (1 + B) \alpha \rfloor + 1)} \middle| (Y_i,W_i)_{i \in \Tr},(Y_i)_{i \in \HO} \right)
    \leq \frac{\lfloor (1 + B) \alpha \rfloor}{1 + B}
    \leq \alpha.
    \qedhere
  \end{align*}
\end{proof}

For the following two proofs, it is helpful to use this version of the Berry--Esseen inequality:

\begin{lem}[A conditional Berry--Esseen inequality]
  \label{lem:conditionalclt}
  Assume that the random variables $X_1,\ldots,X_N \in \R$ are mean-zero and independent conditional on $\mathcal{F}$, that the third (conditional) moments are bounded, and that $\E[X_i^2|\mathcal{F}] > 0$ almost surely for at least one unit $i$.
  Then
  \begin{align*}
    \sup_{t \in \R}\left|\P\left(
      \sum_{i=1}^N X_i \leq t\middle|\mathcal{F}\right)
      -
      \P\left(\N\left(\0,
      \sum_{i=1}^N \Var(X_i|\mathcal{F})
      \right) \leq t \middle|\mathcal{F}\right)\right|
    \leq
    C_0
    \frac{\sum_{i=1}^N \E[|X_i|^3|\mathcal{F}]}{\left(\sum_{i=1}^N \E[X_i^2|\mathcal{F}]\right)^{3/2}} 
  \end{align*}
  almost surely for some universal constant $C_0 < \infty$.
\end{lem}

\begin{proof}[Proof of \autoref{prop:Inference}]
  Conditional unbiasedness is immediate.
  To obtain the conditional central limit theorem, we first consider the case $k=1$, and apply \autoref{lem:conditionalclt} for the hold-out where we condition on $\hat{f}$ as well as $(W_i)_{i \in \HO}$.
  We assume that $|Y_i^{\hat{f}}| \leq c$ a.s.\
  and $\Var\left(Y_i^{\hat{f}}\middle|(Y_i,W_i)_{i \in \Tr}\right) \geq \varepsilon$ a.s.\ for some universal constants $c < \infty, \varepsilon > 0$.
  Starting with
  \begin{align*}
    \sqrt{\underline{h}} (\widehat{\tau} - \tau)
    =
    \sum_{i \in \HO; W_i = 1} \frac{\sqrt{\underline{h}}}{h_1} (Y_i^{\hat{f}} - \E[Y^{\hat{f}}|W{=} 1])
    -
    \sum_{i \in \HO; W_i = 0} \frac{\sqrt{\underline{h}}}{h_0} (Y_i^{\hat{f}} - \E[Y^{\hat{f}}|W{=} 0])
  \end{align*}
  we obtain independent increments
  \begin{align*}
    X_i = \begin{cases}
      +\frac{\sqrt{\underline{h}}}{h_1} (Y_i^{\hat{f}} - \E[Y^{\hat{f}}|W{=} 1]), & W_i=1 \\
      -\frac{\sqrt{\underline{h}}}{h_0} (Y_i^{\hat{f}} - \E[Y^{\hat{f}}|W{=} 0]), & W_i=0
    \end{cases}
  \end{align*}
  with
  \begin{align*}
    &\E[|X_i|^3|\mathcal{F}] \leq \begin{cases}
      \frac{\underline{h}^{3/2}}{h^3_1} 8 c^3, & W_i=1 \\
      \frac{\underline{h}^{3/2}}{h^3_0} 8 c^3, & W_i=0
    \end{cases}
    \leq
    \frac{8 c^3}{\sqrt{\underline{h}}}
    \begin{cases}
      \frac{1}{h_1}, & W_i=1 \\
      \frac{1}{h_0}, & W_i=0
    \end{cases}
    \\
    &\E[X_i^2|\mathcal{F}] 
    =
    \Var(X_i|\mathcal{F})
    =
    \begin{cases}
      \frac{\underline{h}}{h^2_1} \Var(Y^{\hat{f}}|W{=}1), & W_i=1 \\
      \frac{\underline{h}}{h^2_0} \Var(Y^{\hat{f}}|W{=}0), & W_i=0
    \end{cases}
    \geq \varepsilon  \begin{cases}
      \frac{\underline{h}}{h^2_1}, & W_i=1 \\
      \frac{\underline{h}}{h^2_0}, & W_i=0
    \end{cases}
  \end{align*}
  almost surely.
  As a consequence,
  for
  $\Sigma = \sum_{i \in \HO} \Var(X_i|\mathcal{F})  = \frac{\underline{h}}{h_1} \Var(Y^{\hat{f}}|W{=}1) + \frac{\underline{h}}{h_0} \Var(Y^{\hat{f}}|W{=}0)$
  we have that
  $\sum_{i \in \HO} \E[X_i^2|\mathcal{F}] = \Sigma \geq \varepsilon$
  and
  $\sum_{i \in \HO} \E[|X_i|^3|\mathcal{F}] \leq 16 \frac{c^3}{\sqrt{\underline{h}}}$
  and thus
  \begin{align*}
    \sup_{t \in \R}\left|\P\left(
      \sqrt{\underline{h}} (\widehat{\tau} - \tau) \leq t\middle|\hat{f},(W_i)_{i \in \HO}\right)
      -
      \P\left(\N\left(\0,
      \Sigma
      \right) \leq t \middle||\hat{f},(W_i)_{i \in \HO}\right)\right|
    \leq
    C_0
    \frac{16 c^3}{\sqrt{\underline{h}} \varepsilon^{3/2}}.
  \end{align*}
  We obtain the statement in the proposition for $k=1$ by integrating over $(W_i)_{i \in \HO}$ (and pulling the supremum out of the expectation) and letting $\underline{h} \rightarrow \infty$.
  The result for general $k$ follows with the Cram\'er--Wold device.
\end{proof}

\begin{proof}[Proof of \autoref{prop:Combining}]
  Conditional unbiasedness is, again, immediate.
  We obtain the asymptotic statement similarly to the proof of \autoref{prop:Inference}, where we now condition on $\hat{f},\hat{f}^\dagger$ as well as $(W_i,L_i)_{i \in \HO}$ with $L_i = \Ind{i \in \Lb}$ and apply \autoref{lem:conditionalclt} on $\HO$.
  Specifically, we start with $k=1$ and
  \begin{align*}
    \sqrt{\underline{\ell}} (\widehat{\tau}^\dagger - \tau)
    =
    &
    \sum_{i \in \HO; i \in \Lb, W_i = 1}
    \frac{\sqrt{\underline{\ell}}}{\ell_1} \left(
      (Y^{\hat{f}}_i - \E[Y^{\hat{f}}|W{=}1,\hat{f}])
      -
      (1 - \sfrac{\ell_1}{h_1})
      (\widehat{Y}^{\hat{f}}_i - \E[\widehat{Y}^{\hat{f}}|W{=}1,\hat{f},\hat{f}^\dagger])
    \right)
    \\
    &-
    \sum_{i \in \HO; i \in \Lb, W_i = 0}
    \frac{\sqrt{\underline{\ell}}}{\ell_0} \left(
      (Y^{\hat{f}}_i - \E[Y^{\hat{f}}|W{=}0,\hat{f}])
      -
      (1 - \sfrac{\ell_0}{h_0})
      (\widehat{Y}^{\hat{f}}_i - \E[\widehat{Y}^{\hat{f}}|W{=}0,\hat{f},\hat{f}^\dagger])
    \right)
    \\
    &+
    \sum_{i \in \HO; i \notin \Lb, W_i = 1}
    \frac{\sqrt{\underline{\ell}}}{h_1} 
    \left(\widehat{Y}^{\hat{f}}_i - \E[\widehat{Y}^{\hat{f}}|W{=}1,\hat{f},\hat{f}^\dagger]\right)
    \\
    &-
    \sum_{i \in \HO; i \notin \Lb, W_i = 0}
    \frac{\sqrt{\underline{\ell}}}{h_0} 
    \left(\widehat{Y}^{\hat{f}}_i - \E[\widehat{Y}^{\hat{f}}|W{=}0,\hat{f},\hat{f}^\dagger]\right)
  \end{align*}
  for which
  \begin{align*}
    X_i
    =
    \begin{cases}
      + \frac{\sqrt{\underline{\ell}}}{\ell_1} \left(
      (Y^{\hat{f}}_i - \E[Y^{\hat{f}}|W{=}1,\hat{f}])
      -
      (1 - \sfrac{\ell_1}{h_1})
      (\widehat{Y}^{\hat{f}}_i - \E[\widehat{Y}^{\hat{f}}|W{=}1,\hat{f},\hat{f}^\dagger])
      \right), & W_i=1, L_i = 1 \\
      - \frac{\sqrt{\underline{\ell}}}{\ell_0} \left(
      (Y^{\hat{f}}_i - \E[Y^{\hat{f}}|W{=}0,\hat{f}])
      -
      (1 - \sfrac{\ell_0}{h_0})
      (\widehat{Y}^{\hat{f}}_i - \E[\widehat{Y}^{\hat{f}}|W{=}0,\hat{f},\hat{f}^\dagger])
      \right), & W_i=0, L_i = 1 \\
      + \frac{\sqrt{\underline{\ell}}}{h_1}
      \left(\widehat{Y}^{\hat{f}}_i - \E[\widehat{Y}^{\hat{f}}|W{=}1,\hat{f},\hat{f}^\dagger]\right), & W_i=1, L_i = 0 \\
      - \frac{\sqrt{\underline{\ell}}}{h_0}
      \left(\widehat{Y}^{\hat{f}}_i - \E[\widehat{Y}^{\hat{f}}|W{=}0,\hat{f},\hat{f}^\dagger]\right), & W_i=0, L_i = 0
    \end{cases}
  \end{align*}
  and thus
  \begin{align*}
    &\E[|X_i|^3|\mathcal{F}] \leq \begin{cases}
      \frac{\underline{\ell}^{3/2}}{\ell^3_1} 64 c^3, & W_i=1, L_i = 1 \\
      \frac{\underline{\ell}^{3/2}}{\ell^3_0} 64 c^3, & W_i=0, L_i = 1 \\
      \frac{\underline{\ell}^{3/2}}{h^3_1} 8 c^3, & W_i=1, L_i = 0 \\
      \frac{\underline{\ell}^{3/2}}{h^3_0} 8 c^3, & W_i=0, L_i = 0
    \end{cases}
    \leq
    \frac{64 c^3}{\sqrt{\underline{\ell}}}
    \begin{cases}
      \frac{1}{\ell_1}, & W_i=1, L_i = 1 \\
      \frac{1}{\ell_0}, & W_i=0, L_i = 1 \\
      \frac{1}{\max(h_1 - \ell_1,1)}, & W_i=1, L_i = 0 \\
      \frac{1}{\max(h_0 - \ell_0,1)}, & W_i=0, L_i = 0
    \end{cases}
\end{align*}
and $\E[X_i^2|\mathcal{F}]
=
\Var(X_i|\mathcal{F})$ with
\begin{align*}
  \Var(X_i|\mathcal{F})
  =
  \begin{cases}
    \begin{aligned}
      &\frac{\underline{\ell}}{\ell^2_1} \Var(Y^{\hat{f}}|W{=}1,\hat{f})- \frac{2 \underline{\ell}}{\ell_1^2} (1 - \sfrac{\ell_1}{h_1}) \Cov(Y^{\hat{f}},\widehat{Y}^{\hat{f}}|W{=}1,\hat{f},\hat{f}^\dagger) \\
      &+ \frac{\underline{\ell}}{\ell_1^2} (1 - \sfrac{\ell_1}{h_1})^2 \Var(\widehat{Y}^{\hat{f}}|W{=}1,\hat{f},\hat{f}^\dagger)
    \end{aligned}, & W_i=1, L_i = 1 \\
    \begin{aligned}
      &\frac{\underline{\ell}}{\ell^2_0} \Var(Y^{\hat{f}}|W{=}0,\hat{f})- \frac{2 \underline{\ell}}{\ell_0^2} (1 - \sfrac{\ell_0}{h_0}) \Cov(Y^{\hat{f}},\widehat{Y}^{\hat{f}}|W{=}0,\hat{f},\hat{f}^\dagger) \\
      &+ \frac{\underline{\ell}}{\ell_0^2} (1 - \sfrac{\ell_0}{h_0})^2 \Var(\widehat{Y}^{\hat{f}}|W{=}0,\hat{f},\hat{f}^\dagger)
    \end{aligned}, & W_i=0, L_i = 1 \\
    \frac{\underline{\ell}}{h^2_1} \Var(\widehat{Y}^{\hat{f}}|W{=}1,\hat{f},\hat{f}^\dagger), & W_i=1, L_i = 0 \\
    \frac{\underline{\ell}}{h^2_0} \Var(\widehat{Y}^{\hat{f}}|W{=}0,\hat{f},\hat{f}^\dagger), & W_i=0, L_i = 0
  \end{cases}
\end{align*}
almost surely.
As a consequence,
$
  \sum_{i \in \HO} \E[|X_i|^3|\mathcal{F}]
  \leq
  \frac{256 c^3}{\sqrt{\underline{\ell}}}
$
and for
$
  \Sigma^\dagger =
  \sum_{i \in \HO} \E[X_i^2|\mathcal{F}]
  =
  \sum_{i \in \HO} \Var(X_i|\mathcal{F})
$,
\begin{align*}
  \Sigma^\dagger
  =
  &\frac{\underline{\ell}}{\ell_1} \Var(Y^{\hat{f}}|W{=}1,\hat{f})- \frac{2 \underline{\ell}}{\ell_1} \frac{h_1 - \ell_1}{h_1} \Cov(Y^{\hat{f}},\widehat{Y}^{\hat{f}}|W{=}1,\hat{f},\hat{f}^\dagger) \\
  &+ \frac{\underline{\ell}}{\ell_1} \left(\frac{h_1 - \ell_1}{h_1}\right)^2 \Var(\widehat{Y}^{\hat{f}}|W{=}1,\hat{f},\hat{f}^\dagger)
  + \frac{\underline{\ell} (h_1 - \ell_1)}{h^2_1} \Var(\widehat{Y}^{\hat{f}}|W{=}1,\hat{f},\hat{f}^\dagger)
  \\
  &+ \frac{\underline{\ell}}{\ell_0} \Var(Y^{\hat{f}}|W{=}0,\hat{f})- \frac{2 \underline{\ell}}{\ell_0} \frac{h_0 - \ell_0}{h_0} \Cov(Y^{\hat{f}},\widehat{Y}^{\hat{f}}|W{=}0,\hat{f},\hat{f}^\dagger) \\
  &+ \frac{\underline{\ell}}{\ell_0} \left(\frac{h_0 - \ell_0}{h_0}\right)^2 \Var(\widehat{Y}^{\hat{f}}|W{=}0,\hat{f},\hat{f}^\dagger)
  \displaybreak[1]
  \\
  =
  &\frac{\underline{\ell}}{\ell_1} 
  \Big(
  \Var(Y^{\hat{f}}|W{=}1,\hat{f}) 
  -
  2 \frac{h_1 - \ell_1}{h_1} \Cov(Y^{\hat{f}},\widehat{Y}^{\hat{f}}|W{=}1,\hat{f},\hat{f}^\dagger)
  \\
  &+ \frac{h_1 - \ell_1}{h_1} \left( \frac{h_1 - \ell_1}{h_1} + \frac{\ell_1}{h_1} \right) \Var(\widehat{Y}^{\hat{f}}|W{=}1,\hat{f},\hat{f}^\dagger)
  \Big)
  \\
  &+ \frac{\underline{\ell}}{\ell_0}
  \Big(
  \Var(Y^{\hat{f}}|W{=}0,\hat{f})
  -
  2 \frac{h_0 - \ell_0}{h_0} \Cov(Y^{\hat{f}},\widehat{Y}^{\hat{f}}|W{=}0,\hat{f},\hat{f}^\dagger)
  \\
  &+ \frac{h_0 - \ell_0}{h_0} \left( \frac{h_0 - \ell_0}{h_0} + \frac{\ell_0}{h_0} \right) \Var(\widehat{Y}^{\hat{f}}|W{=}0,\hat{f},\hat{f}^\dagger)
  \Big)
  \\
  =
  &\frac{\underline{\ell}}{\ell_1} 
\left( \sfrac{\ell_1}{h_1}
   \Var(Y^{\hat{f}}|W{=}1,\hat{f}) + (1 - \sfrac{\ell_1}{h_1}) \Var(\widehat{Y}^{\hat{f}}-Y^{\hat{f}}|W{=}1,\hat{f},\hat{f}^\dagger) \right)
  \\
  &+ \frac{\underline{\ell}}{\ell_0}
\left( \sfrac{\ell_0}{h_0}
   \Var(Y^{\hat{f}}|W{=}0,\hat{f}) + (1 - \sfrac{\ell_0}{h_0}) \Var(\widehat{Y}^{\hat{f}}-Y^{\hat{f}}|W{=}0,\hat{f},\hat{f}^\dagger) \right)
\end{align*}
which also implies that $\sum_{i \in \HO} \E[X_i^2|\mathcal{F}] = \Sigma^\dagger \geq \varepsilon$.
The result follows as for \autoref{prop:Inference}.
\end{proof}

\end{document}